\RequirePackage{fix-cm}
\documentclass[smallextended]{svjour3}
\journalname{Journal of Mathematical Biology}

\usepackage[sort&compress]{natbib}
\bibpunct{(}{)}{;}{a}{}{,}

\smartqed
\usepackage{epsfig}
\usepackage{graphicx}
\usepackage{amsmath}
\usepackage{amssymb}
\usepackage{amsfonts}
\usepackage{dcolumn}
\usepackage{color}
\newcommand{\R}{\mathbb R}   
\newcommand{\C}{\mathbb C}   

\newtheorem{thm}{Theorem} [section]
\newtheorem{prop}[thm]{Proposition}

\newtheorem{lem}[thm]{Lemma}
\newtheorem{cor}[thm]{Corollary}
\newtheorem{rk}[thm]{Remark}

\catcode`@=11 \@addtoreset{equation}{section} \catcode`@=12

\def\qq#1{\qquad \mbox{#1}\quad}

\newcommand{\al}{\alpha}

\newcommand{\e}{\varepsilon}

\newcommand{\la}{\lambda}
\newcommand{\La}{\Lambda}

\newcommand{\s}{\sigma}
\newcommand{\te}{\theta}

\begin{document}
\title{Waves of cells with an unstable phenotype accelerate the progression of high-grade brain tumors}

\author{Rosa Pardo \and Alicia Mart\'{\i}nez-Gonz\'alez \and V\'{\i}ctor M. P\'erez-Garc\'{\i}a}

\institute{R. Pardo \at
Departamento de Matem\'atica Aplicada,
Facultad de Ciencias Qu\'{\i}micas, Universidad Complutense, Avda. Complutense s/n,
28040 Madrid, Spain  \email{rpardo@mat.ucm.es} \\
A. Mart\'{\i}nez-Gonz\'alez,  V\'{\i}ctor M. P\'erez-Garc\'{\i}a \at Departamento de
Matem\'aticas, E. T. S. I.  Industriales and Instituto de Matem\'atica Aplicada a la Ciencia y la Ingenier\'{\i}a, Universidad de Castilla-La
Mancha, 13071 Ciudad Real, Spain.
\email{victor.perezgarcia@uclm.es, alicia.martinez@uclm.es}}
\date{Received: date / Accepted: date}
\maketitle

\begin{abstract}

In this paper we study  a reduced  continuous model describing the local evolution of high grade gliomas - a lethal type of primary brain tumor - through the interplay of different cellular phenotypes. We show how hypoxic events, even sporadic and/or limited in space may have a crucial role on the acceleration of the growth speed of high grade gliomas. Our modeling approach is based on  two cellular phenotypes one of them being more migratory and the second one more proliferative with transitions between them being driven by the local oxygen values, assumed in this simple model to be uniform. Surprisingly even acute hypoxia events (i.e. very localized in time) leading to the appearance of migratory populations have the potential of accelerating the invasion speed of the proliferative phenotype up to speeds close to those of the migratory phenotype. The high invasion speed of the tumor persists for times much longer than the lifetime of the hypoxic event and the phenomenon is observed both when the migratory cells form a persistent wave of cells located on the invasion front and when they form a evanecent wave dissapearing after a short time by decay into the more proliferative phenotype.

Our findings are obtained through numerical simulations of the model equations. We also provide a deeper mathematical analysis of some aspects of the problem such as the conditions for the existence of persistent waves of cells with a more migratory phenotype.


\end{abstract}

\section{Introduction}

\label{intro}

Malignant gliomas are the most frequent type of primary brain tumor. Between them, the most aggressive and prevalent  glioma in adults is the glioblastoma multiforme (GBM), a grade IV astrocytic tumor \citep{Wen}. Mean survival after GBM diagnosis is around 14 months using the standard of care which includes surgery to resect as much tumoral tissue as possible, radiotherapy and chemotherapy (temozolamide) \citep{Mangiola2010}. Despite advances in understanding the complex biology of these tumors, the overall prognosis has  improved only slightly in the past three decades.

The main reason for treatment failure is that the periphery of the GBM typically shows tumor cells infiltrating into the normal brain, frequently even in the contralateral hemispherium. Thus, even the so-called gross total resection of the tumor does not eliminate many migrating cells that cause tumor recurrence \citep{Onishi,Berens}, typically in less than six months after surgery \citep{Giese}.

Pathological features of GBM are cellular pleomorphism, high cellular proliferation and diffuse infiltration, necrosis in the central regions of the tumor, microvascular hyperplasia and hypercellular areas surrounding necrotic areas around thrombosed vessels called pseudopalisades.

Due to the abnormal cell proliferation, the pre-existing vascular network is not able to appropriately feed the tumor cells. Angiogenesis emerges then in response to proangiogenic growth factors that are released by hypoxic cells in the tumor such as vascular endothelial growth factor (VEGF) \citep{Ebos}. The end result of VEGF signaling in tumors is the production of immature, highly permeable blood vessels with subsequent poor maintenance of the blood brain barrier and parenchymal edema \citep{Jensen} leading to hypoxia.

Tumor hypoxia is generally recognized as a negative clinical prognostic and predictive factor owing to its involvement in various cancer hallmarks such as resistance to cell death, angiogenesis, invasiveness, metastasis, altered metabolism and genomic instability \citep{Ranalli2009,Hanahan2011}. Hypoxia plays a central role in tumor progression and resistance to therapy (chemo- and radioresistance), specially in GBM, where
it has been proven to play a key role in the biology and aggression of these cancers \citep{Evans2004}. These facts have motivated considering hypoxia as a therapeutic target in cancer. Some hypoxia-regulated molecules, including hypoxia inducible factor-1 (HIF-1), carbonic anhydrase IX, glucose transporter 1, and VEGF, may be suitable targets for therapies.  HIF-1 is a regulator transcriptor of cell adaptation to hypoxia efficiently translated under normoxic and hypoxic conditions, however subunit HIF-1$\alpha$ contains an oxygen-dependent degradation domain 
which is rapidly degraded in normoxia.

Endothelial injury and prothrombotic factors secreted  by glioma cells \citep{Dutzmann2010} lead to vaso-occlusive events. Following these events  necrotic regions and waves of hypoxic cells moving away from the perivascular anoxic regions \citep{pseudopalisading2,vaso-occlusive,Alicia2012} are generated. Cells located in the perivascular areas have both oxygen and nutrients and have a high proliferative activity. However, cells exposed to hypoxia display increased migration and slower proliferation to deal with a more aggressive environment \citep{Zagzag2000,Elstner2007,Das2008}. This phenomenon has been called the \emph{go-or-grow dichotomy} and studied in detail in gliomas \citep{Giese1998,Giese}.
  

Many mathematical papers  have used the concept of the migration-proliferation dichotomy to explain different aspects of the behavior of tumor cell populations in vitro or in vivo \citep{Iomin,Stein2007,Fedotov2,Pham2011,Tektonidis2011,Hatzikirou2012,Alicia2012,Alicia2014}.
Specifically, several works have considered the role of hypoxia in gliomas finding a potential beneficial effect of its reduction via either the
increase of the oxygen tension in the tumour \citep{Hatzikirou2012} or by the reduction of the occurrence of thrombotic events \citep{Alicia2012,Alicia2014}.

 In this work we complement previous studies and show that despite the initial idea of improving oxygenation is reasonable and well founded and may lead to better response to therapies, the effects of hypoxia may be more perverse than initially considered in previous works since \emph{even minimal amounts of hypoxic events may lead to accelerated progression in gliomas, even when oxygenation is rapidly restored and persistent for very long times.}

The plan of the paper is as follows. First, in Sec. \ref{model} we present the mathematical model, some preliminary theoretical results and discuss the parameter ranges of interest. In Sec. \ref{results} we present the results of our numerical simulations showing the acceleration of invasion by \emph{waves of cells with an hypoxic phenotype}. A detailed theoretical study with some rigorous results on travelling waves of the system under study is developed in Sec. \ref{Theoretical}. Finally in Sec. \ref{discussion} we discuss the practical implications of our findings and summarize our results.

\section{The Model}\label{model}
\par

\subsection{Derivation of the Model}\label{Derivation}
\par

Following the go-or-grow dichotomy concept we will describe the tumor as expressing two different phenotypes: a proliferative one to be  denoted as $u_n(x,y)$ and a migratory one $u_h(x,y)$. We will consider that the force driving phenotype changes is the local oxygen pressure so that in hypoxic conditions tumor cells change to a mobile phenotype in a characteristic time $\tau_{nh}$ and in normoxic conditions tumor cells acquire a proliferative phenotype in a time $\tau_{hn}$.

Detailed models of these processes have been proposed in several papers (see e.g. \citet{Alicia2012,Alicia2014}).
In this paper we will use a minimal model intended to capture the essentials of a striking phenomenon, the acceleration of tumor invasion due to sporadic hypoxic events.

We will assume that an initial hypoxic event around a capilar lasts a sufficient time for a complete phenotype switch to the migratory phenotype of the tumor cells located around it. This is reasonable because of the fast response of HIF-1 under hypoxia, that induces phenotypic changes in a characteristic time of the order of minutes \citep{Jewell2001}.

Thus, we will take initially our tumor density to be of the form $u_n(x,t=0) = 0, u_h(x,t=0) = u_h^0(x)$ and localized around a tumor vessel. Once oxygen supply is restored we will assume that oxygenation is maintained above the hypoxia level for all times, may be due to the action of a therapy normalizing vasculature, avoiding thrombotic events and/or increasing oxygenation.
In that scenario the dynamics will be described by the equations
\begin{subequations}
\label{themodel}
\begin{eqnarray}
\dfrac{\partial u_n}{\partial t} = D_n \dfrac{\partial^2 u_n}{\partial x^2} + \dfrac{1}{\tau_n} \left(1-u_n-u_h\right) u_n + \dfrac{1}{\tau_{hn}} u_h, \label{themodel1} \\
\dfrac{\partial u_h}{\partial t} = D_h \dfrac{\partial^2 u_h}{\partial x^2} + \dfrac{1}{\tau_h} \left(1-u_n-u_h\right) u_h - \dfrac{1}{\tau_{hn}} u_h, \label{themodel2}
\end{eqnarray}
\end{subequations}
where $D_n, D_h$ are the diffusion coefficients for the normoxic (proliferative) and hypoxic (migratory) phenotypes satisfying $D_h > D_n$ and $\tau_n, \tau_h$ are the doubling times for both phenotypes.
This model is a pair of coupled  Fisher-Kolmogorov equations including a coupling term accounting for the decay of  hypoxic cells  into the normoxic phenotype with a characteristic time $\tau_{hn}$.

\subsection{Global existence and boundedness of model's solutions}\label{Preliminary}
\par
Let us first study  the problem of global existence in time of non-negative solutions of Eqs. \eqref{themodel}  with initial data $ u_n(x,0)=u_n^0(x)$, $u_h(x,0)=u_h^0(x)$ and $x\in\R$, $t>0$. We will   assume, in agreement with their biological meaning, that $D_n,D_h,\tau_n,\tau_h,\tau_{hn} >0$ are finite real parameters.

\begin{thm} For any   $\ (u_n^0,u_h^0)\in L^\infty\cap H^1(\R)^2$
\begin{enumerate}
\item[\rm (i)] There exists a time $T=T(u_n^0,u_h^0)$ such that the parabolic problem \eqref{themodel} has a unique solution $ (u_n(x,t),\  u_h(x,t))$ for $0<t<T .$

\item[\rm (ii)] Any non-negative solution  $\ (u_n(x,t), u_h(x,t))$ of \eqref{themodel}, is a classical solution defined globally in time, and there are constants $M_1,$ $M_2$ such that
\begin{subequations}
\begin{equation}
 0\leq u_n(x,t)\leq M_1:=\max \left\{\|u_n^0\|_\infty,\frac{1}{2}\left( 1+\sqrt{1+4M_2 \tau_n/\tau_{hn}}\right)\right\},  \label{def:Mu}
 \end{equation}
 \begin{equation}
 0\leq u_h(x,t)\leq  M_2:=\max\left\{\|u_h^0\|_\infty,\ 1- \tau_h/\tau_{hn} \right\}.\label{def:Mw}
\end{equation}
\end{subequations}

Moreover, if $\tau_{hn}<\tau_{h}$, then $u_h(x,t)\to 0$ as $t\to\infty.$
\end{enumerate}
\end{thm}

\begin{proof}. (i) For initial data $\ (u_n^0,u_h^0)\in  H^1(\R)^2  $, the existence and uniqueness of mild solutions in $C([0,T); L^2(\R)^2)$ holds by the variation of constants formula and standard fixed points arguments.
We take $X=L^2(\R)^2$ and $A$ is the closure in $X$ of the differential operator
$A(u_n,u_h):=\left(-D_n \frac{\partial^2 u_n}{\partial x^2},-D_h \frac{\partial^2 u_h}{\partial x^2}\right)
$
in $C_0^\infty (\R)^2.$ Let us consider the fractional power spaces $X^\alpha:=D\left((A+I)^\alpha\right)$, $\alpha\ge 0$  (see \citet{Henry}). In particular $X^{1/2}=H^1(\R)^2$ and $X^1=H^2(\R)^2$.

Let $u=(u_n,u_h)$ , $v=(v_n,v_h)$ and $f(u)=(f_1,f_2),$ be given by
\begin{subequations}
\label{nlin}
\begin{eqnarray}
f_1(u_n,u_h) & := &\left(1-u_n-u_h\right) u_n /\tau_n+ u_h/\tau_{hn},  \\
f_2(u_n,u_h) & := & \left(1-u_n-u_h\right) u_h/\tau_h -  u_h/\tau_{hn},
\end{eqnarray}
\end{subequations}
If $1>\al>1/4$,  we have
\begin{eqnarray*}
    \|f(u)-f(v)\|_{L^2(\R)^2} &\leq & C\left(\|u\|_{L^\infty(\R)^2} +\|v\|_{L^\infty(\R)^2}\right) \|u-v\|_{L^2(\R)^2} \\
     &\leq& C\left(\|u\|_{\alpha} +\|v\|_{\alpha}\right) \|u-v\|_{\alpha},
 \\
    \|f(u)\|_{L^2(\R)^2} &\leq& C\|u\|_{L^\infty(\R)^2} \|u\|_{L^2(\R)^2}
\leq C\|u\|_{\alpha}^2
  \end{eqnarray*}
so the hypothesis of  \citet[Theorem 3.3.3]{Henry} are verified and local existence and uniqueness follows.

(ii)  Next, we will prove global existence for nonnegative solutions.
To get the global bounds, we use standard comparison arguments. For any non-negative solution we have that
\begin{equation}
\dfrac{\partial u_h}{\partial t} \leq D_h \dfrac{\partial^2 u_h}{\partial x^2}+\dfrac{1}{\tau_{h}}(1-u_h)u_h-\dfrac{1}{\tau_{hn}}\, u_h.
\end{equation}
Therefore, since $0\leq u_h^0(x)\leq \|u_h^0\|_\infty$
we obtain $0\leq u_h(x,t)\leq V_h(t)$,
where $V_h(t)$ is the solution of the problem
\begin{equation}
\dfrac{dV_h}{dt}=\dfrac{1}{\tau_{h}}\left[\left(1-\dfrac{\tau_h}{\tau_{hn}}\right)- V_h\right]V_h,\qquad V_h(0)=\|u_h^0\|_\infty.
\end{equation}
Let us choose $u_h^0\gneqq 0$. Obviously,  $V_h(0)> 0.$
If $1\le \tau_h/\tau_{hn}$, then $dV_h/dt\leq 0$, $V_h(t)$ will be a decreasing function on $[0,\infty),$ and moreover $V_h(t)\to 0$ as $t\to\infty.$ Consequently $u_h(x,t)\to 0$ as $t\to\infty.$

Assume now $1> \tau_h/\tau_{hn}.$ If $V_h(0)=1-\tau_h/\tau_{hn}$ then $dV_h/dt= 0$ and therefore $V_h(t)=1-\tau_h/\tau_{hn}$  on $[0,\infty).$
If $V_h(0)<1-\tau_h/\tau_{hn}$ then $dV_h/dt> 0$ and therefore $V_h(t)$ will be an increasing function on $[0,\infty),$ and upper bounded,
consequently $V_h(t)\leq 1-\tau_h/\tau_{hn}$  on $[0,\infty).$
If $V_h(0)>1-\tau_h/\tau_{hn}$ then $dV_h/dt \leq 0$ and therefore $V_h(t)$ will be a decreasing function on $[0,\infty),$ and lower bounded,
consequently $V_h(t)\leq V_h(0)$  on $[0,\infty),$
and \eqref{def:Mw} holds.
Moreover, for any positive initial data, $V_h(t)\to 1-\tau_h/\tau_{hn}$ as $t\to\infty.$
\medskip

For $u_n$ we get
\begin{equation}
\dfrac{\partial u_n}{\partial t} \leq D_n \dfrac{\partial^2 u_n}{\partial x^2}+\dfrac{1}{\tau_{n}}(1-u_n)u_n+\dfrac{M_2}{\tau_{hn}},
\end{equation}
where $M_2$ is given by \eqref{def:Mw}. Therefore, as in the previous case $0\leq u_n(x,t)\leq V_n(t)$,
where $V_n(t)$ is the solution of the problem
\begin{equation}
\dfrac{dV_n}{dt}=\dfrac{1}{\tau_{n}}(1-V_n)V_n+\dfrac{M_2}{\tau_{hn}},\qquad V_n(0)=\|u_n^0\|_\infty.
\end{equation}

Let us choose $u_n^0\gneqq 0$, thus $V_n(0)> 0.$  If $V_n(0)=1/2+\sqrt{1/4+\left(M_2\, \tau_n\right)/\tau_{hn}}$ then $dV_n/dt= 0$ and therefore $V_n(t)=V_n(0)$  on $[0,\infty).$
If $V_n(0)< 1/2 +\sqrt{1/4+\left(M_2\, \tau_n\right)/\tau_{hn}}$ then $dV_n/dt> 0$ and therefore $V_n(t)$ will be an increasing function on $[0,\infty),$ and upper bounded,
consequently $V_n(t)\leq 1/2 +\sqrt{1/4+\left(M_2\, \tau_n\right)/\tau_{hn}} $  on $[0,\infty).$
If $V_n(0)>\frac12 +\sqrt{1/4+\left(M_2 \tau_n\right)\tau_{hn}}  $ then $dV_n/dt\leq 0$ and therefore $V_n(t)$ will be a decreasing function on $[0,\infty),$ and lower bounded,
consequently $V_n(t)\leq V_n(0).$ Moreover, for any positive initial data, $V_n(t)\to 1/2+\sqrt{1/4+\left(M_2\, \tau_n\right)/\tau_{hn}}$ as $t\to\infty.$
Hence,  inequality \eqref{def:Mu} holds, which ends the proof.
\end{proof}

\begin{prop}
Assume that $\tau_{hn}>\tau_{n}.$ The box
\begin{equation}
\Sigma :=\{(u_n,u_h)\ :\ 0\leq u_n,u_h\leq 1\}
\end{equation}
is an invariant box for  \eqref{themodel}, i.e. for any initial data $(u_n^0,u_h^0) \in \Sigma,$
then the solution of \eqref{themodel} $\ (u_n(x,t),u_h(x,t)) \in \Sigma\ $ for any $t\geq 0$.
\end{prop}
\begin{proof} Taking $f$ as given by Eqs. (\ref{nlin}), assuming that $(u_n,u_h)\in\Sigma ,$ and that $\tau_{hn}>\tau_{n},$ then
\begin{subequations}
\begin{eqnarray}
 & f_1(0,u_h) := u_h/\tau_{hn}\geq 0,\qquad &
f_1(1,u_h) := \left(1/\tau_{hn}-1/\tau_{n}\right)u_h\leq 0, \\
& f_2(u_n,0):=0,\qquad &
f_2(u_n,1):=-u_n/\tau_{h}-1/\tau_{hn}\leq 0.
\end{eqnarray}
\end{subequations}
Therefore, $\Sigma $ is an invariant box, see \citet[Section 14.B, p. 198]{Smoller}.
\end{proof}

\subsection{Parameter estimation}\label{Parameter}
\par

Brain tissue has a complex structure with spatial inhomogeneities in the parameter values (e.g. different propagation speeds in white and gray matter) and anisotropies (e.g. on the diffusion tensor with preferential propagation directions along white matter tracts). In order to simplify the analysis and focus on the essentials of the phenomena to be studied we have chosen to study the model in one spatial dimension and in isotropic media.

In Eqs. (\ref{themodel}), the densities $u_n$, $u_h$ are taken in units of the maximal tissue density, typically around   $10^3$ cell/cm.  Although we will explore different parameter regimes, the normoxic cell doubling time will be taken to be $\tau_n \sim 24 $ h in agreement with typical cell doubling times in vitro (see e.g. \citet{doubling}) and the hypoxic one $\tau_h \sim 48$ h. The diffusion coefficients for normoxic and hypoxic cells will be taken             to be around
$D_n = 6.6 \cdot 10^{-12}$ cm$^2$/s and $D_h$ around an order of magnitude larger \citep{Alicia2012}. The parameter
$\tau_{hn}$, corresponding to the phenotype switch time is harder to estimate since it corresponds to the recovery of the less motile proliferative phenotype
 under conditions of good oxygenation. The time of the opposite transition $\tau_{nh}$ corresponding to the response time to hypoxia, despite some variability \citep{Chi2006}, is a very fast time \citep{Jewell2001} corresponding to the fast activation of the cellular adaptive responses to match oxygen supply with metabolic, bioenergetic, and redox demands \citep{Majmundar2010}. Normal cells have the capability of restoring their normal behavior once a hypoxic stimulus has finished. However, cancer cells may become more aggressive after cycles of hypoxia and reoxigenation \citep{Bashkara2012}, and in some aspects the process may become at some moment irreversible leading to the so-called Warburg phenotype \citep{Koppenol,Berta2012}.  Typically short cycles of about 30 min of hypoxia and reoxigenation lead to the same (or even worse) outcome than under chronic hypoxia \citep{FEBS} leading to the conclusion that $\tau_{hn}$ is larger than this value and thus $\tau_{hn} \gg \tau_{nh}$. Different works point out to a normalization of the response to hypoxia between 48 h and 72 h. In vivo analysis of HIF-1$\alpha$ stabilization in well oxygenated tumor areas \citep{Zagzag2000}
 provides support for long normalization times $\tau_{hn}$.

 To solve Eqs. \eqref{themodel} numerically we have used a standard finite difference method of second order in time and space with zero boundary conditions on the boundaries of the integration domain. We have used large integration domains and cross-checked our results for different sizes to avoid boundary effects.

 \section{Numerical results}
 \label{results}

\subsection{Hypoxic events lead to fast glioma progression}

\begin{figure}
\begin{center}
\includegraphics[width=0.60\textwidth]{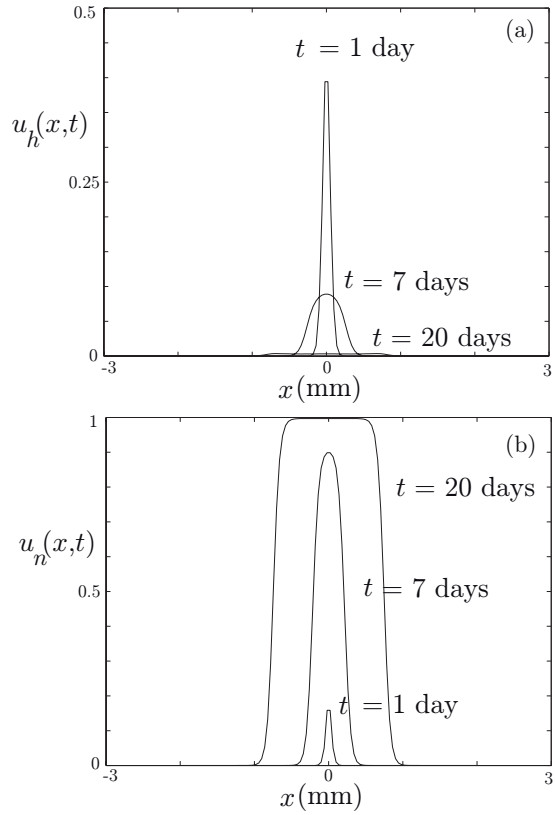}
\caption{Snapshots of the evolution of the (a) hypoxic  $u_h(x,t)$ and (b) normoxic $u_n(x,t)$ densities solving Eq. (\ref{themodel}) for parameter values $D_n = 6.6\times 10^{-12}$ cm$^2$s$^{-1}$, $D_h = 1.32\times 10^{-10}$ cm$^2$s$^{-1}$, $\tau_n =$ 24 h, $\tau_h=$ 48 h, $\tau_{hn} = $ 96 h and initial data $u_n(x,0) = 0$, $u_h(x,0) = \left(1 - 800 x^2\right)_+$, $x$ being measured in mm. Shownare the results for $t=1, 7$ and 20 days showing the decay of the hypoxic cell density and the formation of a front of normoxic cells. \label{snapshots}}
\end{center}
\end{figure}

An example of the phenomenon to be described here is presented in Figs. \ref{snapshots} and  \ref{prima}.
There, an initial distribution of hypoxic tumor cells at $t=0$, e.g. due to a transient vaso-occlusive event, is placed in a well oxygenated environment.
One might expect naively that after a transient of about a few times $\tau_{hn}$, hypoxic cells would dissapear and then the front speed would tend asymptotically to that of the normoxic phenotype, given by the minimal FK speed
\begin{subequations}
\begin{equation}\label{def:cn*}
c^*_n = 2 \sqrt{D_n/\tau_{n}}
\end{equation}
 Although the hypoxic cell density decays and a front of normoxic cells is generated (see Fig. \ref{snapshots}), the normoxic front propagation speed is not the minimal FK speed. In Fig. \ref{prima} we show pseudocolor spatio-temporal plots of $u_n(x,t)$ (Fig. \ref{prima}(a)) and the front speed (Fig. \ref{prima}(b)). The internal and external white lines show (respectively) the predicted evolution of purely normoxic and hypoxic initial data under no phenotype changes (i.e., the limit $\tau_{hn} \rightarrow \infty$).

\begin{figure}
\begin{center}
\includegraphics[width=0.72\textwidth]{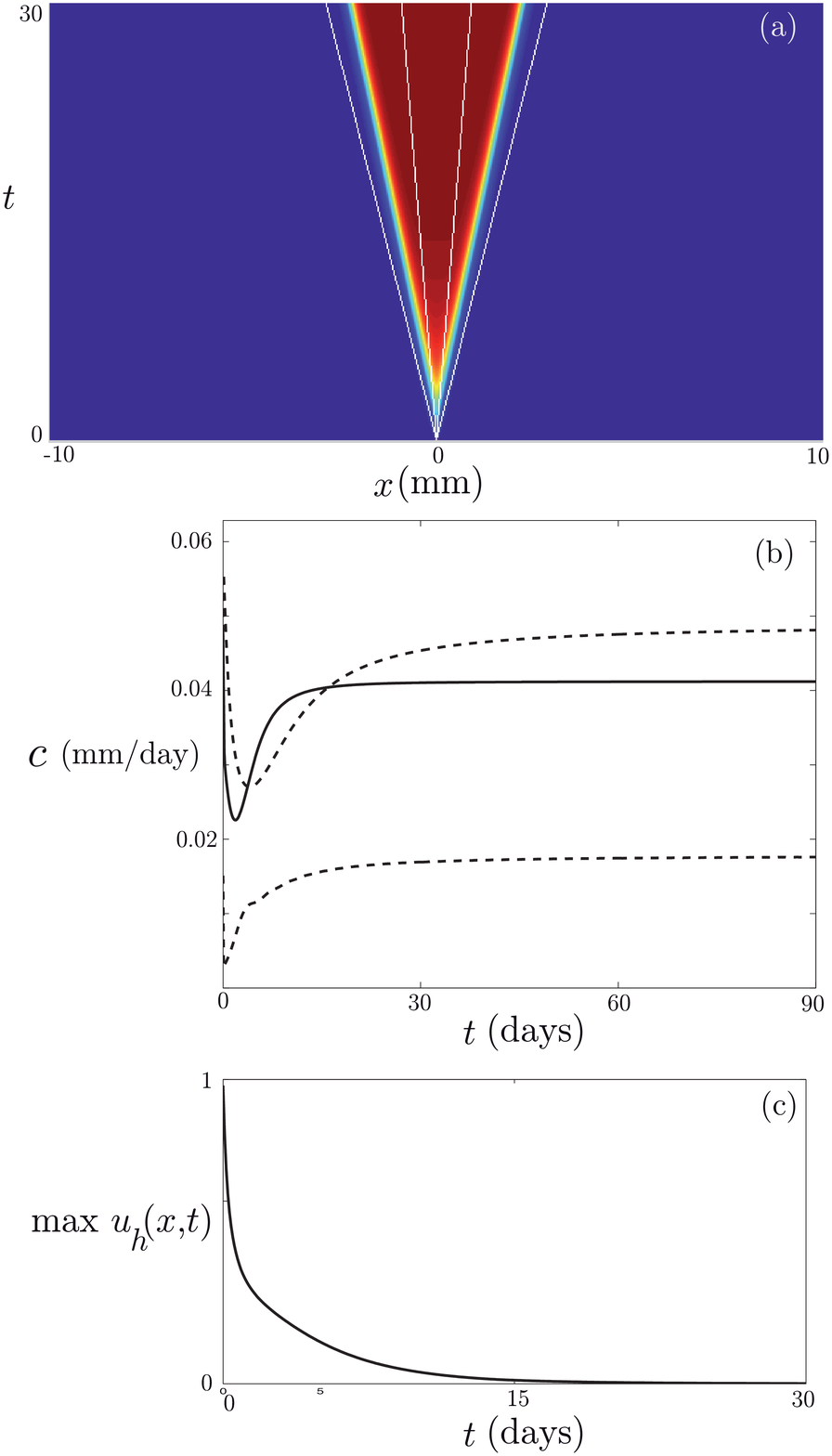}
\caption{Evolution of the normoxic  $u_n(x,t)$ and hypoxic $u_h(x,t)$ cell densities solving Eq. (\ref{themodel}) for parameter values $D_n = 6.6\times 10^{-12}$ cm$^2$s$^{-1}$, $D_h = 1.32\times 10^{-10}$ cm$^2$s$^{-1}$, $\tau_n =$ 24 h, $\tau_h=$ 48 h, $\tau_{hn} = $ 96 h and localized initial data of the form $u_n(x,0) = 0$, $u_h(x,0) = \left(1 - 800 x^2\right)_+$, $x$ being measured in mm. (a) Pseudocolor plots of the amplitude of the normoxic cell density $u_n(x,t)$. The white lines are calculated using the propagation speed of the hypoxic and normoxic fronts in the framework of the FK equation (b) Evolution of the normoxic wavefront speed calculated as $v(t) = \dfrac{d}{dt} \left[2\left( \int \left| x - X(t)\right| u_n(x,t) dx\right)/\left(\int u_n(x,t) dx\right)\right] $, where $X(t) = \int x u_n(x,t) dx$. Shownare for comparison the speeds (calculated using the same formula) for purely hypoxic (upper dashed line) and purely normoxic cells (dashed lower line) without any type of transitions allowed between them. The long time scales shown guarantee that asymptotically the propagation of the normoxic front under an initial hypoxic event reaches an asymptotic speed essentially larger than the purely normoxic front.  (c) Decay of the amplitude of the hypoxic component $\max_x u_h(x,t)$, as a function of time.\label{prima}}
\end{center}
\end{figure}

From the results shown in  Fig. \ref{prima} it is clear that the speed of the propagating front of normoxic cells \emph{is not $c^*_n$ but given by a substantially larger value even when the hypoxic cell population has become extinct} (cf. Fig. \ref{prima}(c)). The simulation runs for long times to show that, although the hypoxic cell amplitude decays in a few days, its effect on the normoxic front propagation speed persists and provides a sustained front acceleration that is still present (though minimal) after one month of the initial ($t=0^-$) hypoxic event.

\begin{figure}
\begin{center}
\includegraphics[width=0.72\textwidth]{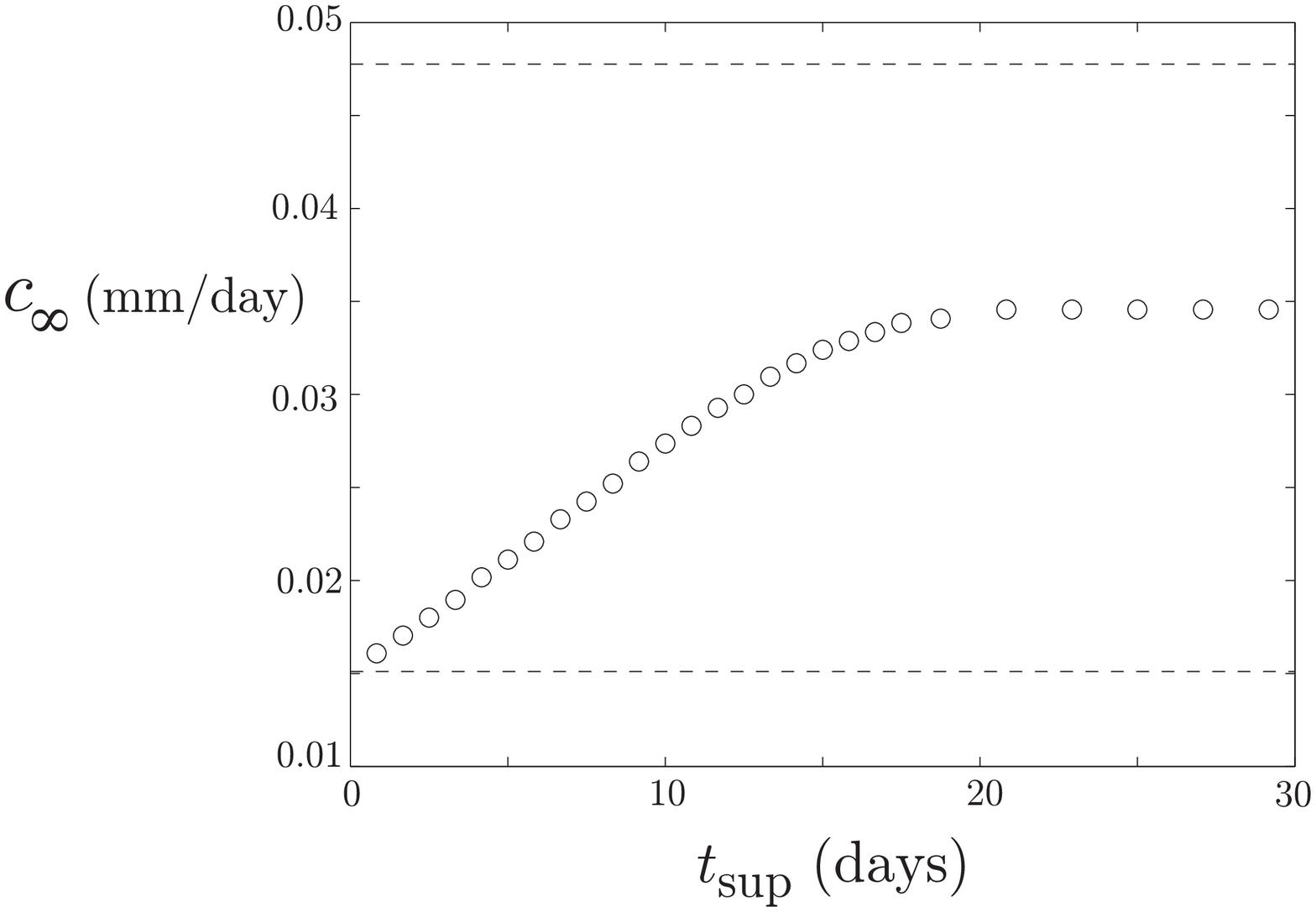}
\caption{(a) Effect on the final asymptotic speed at $t=30$ days of the suppression of the hypoxic component for different times $t_{sup}$. The horizontal dashed lines denote the normoxic $c_n$ (lower line) and hypoxic $c_h$ (upper line) phenotype speeds respectively. (b) Decay of the amplitude of the hypoxic cell subpopulation. Parameter values are $D_n = 6.6\times 10^{-12}$ cm$^2$s$^{-1}$, $D_h = 1.32\times 10^{-10}$ cm$^2$s$^{-1}$, $\tau_n =$ 24 h, $\tau_h=$ 48 h, $\tau_{hn} = $ 48 h and initial data as in Fig. \ref{prima}. \label{tertia}}
\end{center}
\end{figure}

\subsection{Waves of hypoxic cells drive the evolution of the normoxic component}

To confirm that the acceleration of the normoxic phenotype front is due to the presence of the hypoxic component we have performed a second set of numerical experiments. In these simulations we have studied the effect of the suppression of the hypoxic component at a given time on the final propagation speed. Thus we have solved Eqs. (\ref{themodel}) but setting artificially $u_h(x,t_{sup}) = 0$ for different values of $t_{sup}$.

The results are summarized in Fig. \ref{tertia}. Note that for the simulations of Fig. \ref{tertia} we have chosen a substantially smaller $\tau_{hn}$ leading to a faster decay of the amplitude of the hypoxic component.
It is remarkable that the asymptotic velocity increases with a longer presence of hypoxic cells even when the hypoxic wave amplitude is very small,
 it continues accelerating the front of normoxic phenotype cells. This behavior is somehow counterintuitive implying that the final normoxic wave speed more than doubles the expected speed
 $c_n$.

To understand this phenomenon let us note that initially the hypoxic wave propagates with speed
\begin{equation}\label{def:ch*}
c^*_h = 2 \sqrt{D_h/\tau_{h}}
 \end{equation}
 \end{subequations}
with $c^*_h > c^*_n$ while at the same time decreases its amplitude providing a seeding of hypoxic cells that later change their phenotype to normoxic thus providing an extra source for normoxic cells. It is interesting that the phenomenon is mediated by a small amplitude wave of hypoxic cells.
This leads to a substantially faster growth of the normoxic component when the wave reaches a specific location. Thus, the outcome of this dynamical phenomenon is that \emph{the effects of an initial hypoxic event may have a substantial influence on the invasion wave speed for very long timescales}.

One may wonder if the outcome of our simulations may be due to the fact that the hypoxic initial distribution is not surrounded by normoxic tumor cells as it may happen in real tissue. To rule out this possibility we have run simulations with initial data localized around a vessel but surrounded by normoxic cells.
An example of our results is shown in Fig. \ref{hyphop} ruling out the influence this choice on the asymptotic speed.

\begin{figure}
\begin{center}
\includegraphics[width=0.72\textwidth]{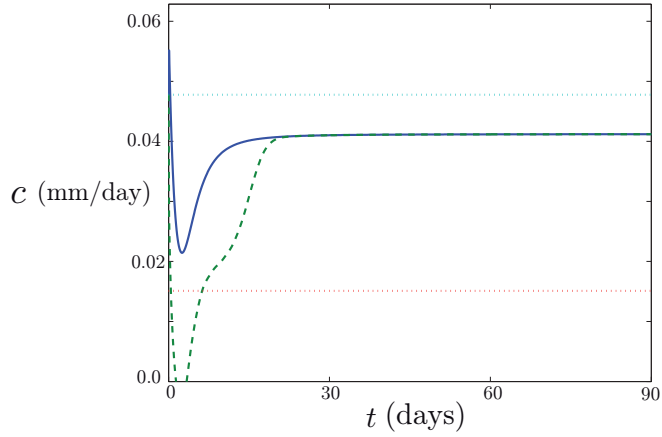}
\caption{Tumor growth speeds  for parameter values $D_n = 6.6\times 10^{-12}$ cm$^2$s$^{-1}$, $D_h = 1.32\times 10^{-10}$ cm$^2$s$^{-1}$, $\tau_n =$ 24 h, $\tau_h=$ 48 h, $\tau_{hn} = $ 96 h and different initial data. The solid line corresponds to the same situation as in Fig. \ref{prima}, i.e. initial data of the form $u_n(x,0) = 0$, $u_h(x,0) = \left(1 - 800 x^2\right)_+$, $x$ being measured in mm. The dashed lines correspond to the evolution of initial data of the form $u_n(x,0) = 0.72\left[1-200(x-0.2)^2\right]_+ + 0.72\left[1-200(x+0.2)^2\right]_+ $, $u_h(x,0) = \left(1 - 800 x^2\right)_+$ corresponding to broad normoxic tumor cell densities around vessels adjacent to the one failing and generating the localized hypoxic density $u_h$. Shownare for comparison the asymptotic speeds of purely hypoxic (upper dotted line) and purely normoxic cells (lower dotted line) calculated using the propagation speed of the hypoxic and normoxic fronts in the framework of the standard scalar FK equation.
\label{hyphop}}
\end{center}
\end{figure}

It is important to emphasize that the phenomenon described in this section occurs in broad parameter regions including the biologically relevant ranges of parameters. As an example, in Fig. \ref{vartaun} we plot the asymptotic speed as a function of the normoxic doubling time ranging from the typical in vitro values of $\tau_n = 24$ h to larger values closer to in-vivo doubling time estimates \citep{Wang,Kirkby}

\begin{figure}
\begin{center}
\includegraphics[width=0.72\textwidth]{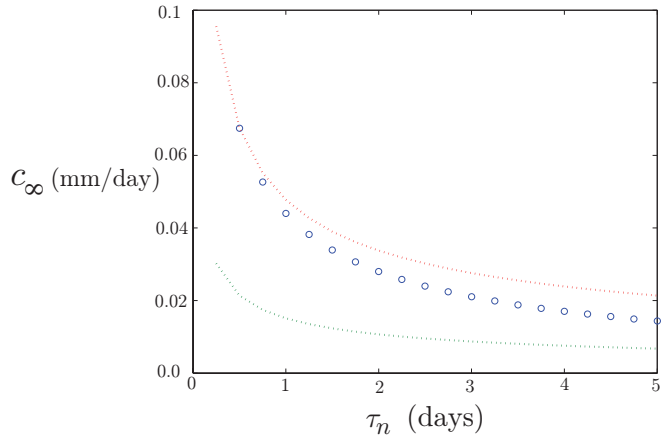}
\caption{Asymptotic tumor growth speed as a function of the normoxic proliferation time $\tau_n$. The hypoxic proliferation time is chosen to be twice the value of $\tau_n$, i.e. $\tau_h = 2\tau_n$.
The dashed lines denote the normoxic $c^*_n$ (lower dashed line) and hypoxic $c^*_h$ (upper line) phenotype asymptotic speeds respectively. Parameter values are $D_n = 6.6\times 10^{-12}$ cm$^2$s$^{-1}$, $D_h = 1.32\times 10^{-10}$ cm$^2$s$^{-1}$, $\tau_{hn} = $ 7 days and initial data as in Fig. \ref{prima}. \label{vartaun}}
\end{center}
\end{figure}

\subsection{Propagation regimes as a function of $\tau_{hn}$}

The parameter $\tau_{hn}$ describes complicated biological processes since it corresponds to the phenotype change of hypoxic cells under oxic conditions. As described in Sec. \ref{model}, this parameter can range from about an hour in normal cells to many days in transformed cells. One would expect that very fast decaying hypoxic phenotypes would lead to a smaller acceleration of the normoxic wave of invasion. In the opposite limit it is to expected that the asymptotic front velocity $c_{\infty}$ satisfies $\lim_{\tau_{hn} \rightarrow \infty} c_{\infty}(\tau_{hn}) = c^*_h$.

\begin{figure}[h]
\begin{center}
\includegraphics[width=0.72\textwidth]{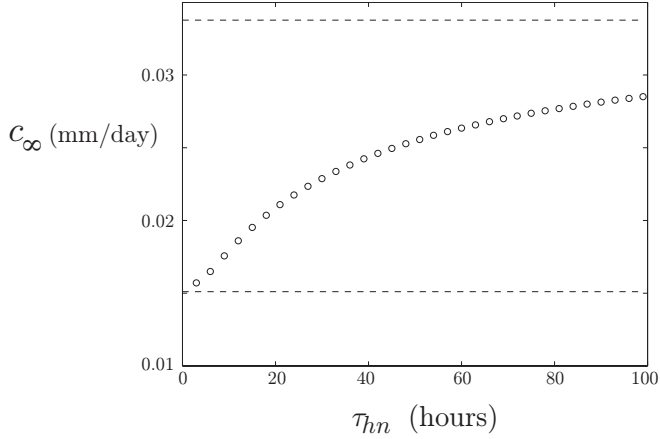}
\caption{Asymptotic speed $c_{\infty}$ of the solutions of Eqs. (\ref{themodel}) after one month  as a function of the switch time $\tau_{hn}$.  Parameter values $\tau_n = 24$ h , $\tau_h = 48 $h, $D_n = 6.6 \times 10^{-12}$ cm$^2$s$^{-1}$, $D_h =$ 6.6 $\times 10^{-11}$ cm$^2$s$^{-1}$ and initial data as in Fig. \ref{prima}. The dashed lines represent the minimal asymptotic speeds of hypoxic  $c^*_h$ (upper line) and normoxic $c^*_n$ invasion waves (lower line). \label{cuarta}}
\end{center}
\end{figure}

In Fig. \ref{cuarta} we explore the behavior of the asymptotic velocity of the normoxic wave one month after the hypoxic event for different values of the switch parameter $\tau_{hn}$. It can be seen how larger recovery times lead to asymptotic speeds closer to $c^*_h$.

\subsection{Localized finite-amplitude wavepackets of hypoxic cells}\label{goandgrow}
\par

It is very interesting that despite the instability of the hypoxic cells there are parameter regimes in which the hypoxic cells do not form an evanescent wave but instead a finite amplitude wavepacket of cells with the hypoxic phenotype persists leading the advance of the tumor.

\begin{figure}[h!]
\begin{center}
\includegraphics[width=0.72\textwidth]{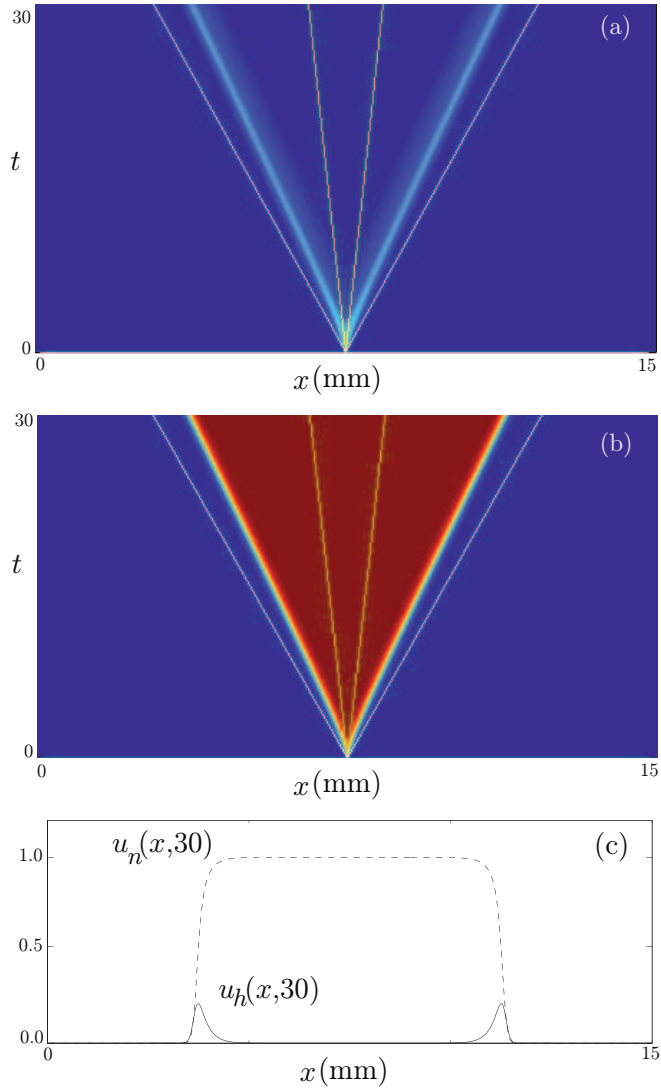}
\caption{Evolution of the normoxic  $u_n(x,t)$ and hypoxic $u_h(x,t)$ cell densities solving Eq. (\ref{themodel}) for parameter values $D_n = 6.6\times 10^{-12}$ cm$^2$s$^{-1}$, $D_h = 1.32\times 10^{-10}$ cm$^2$s$^{-1}$, $\tau_n =$ 24 h, $\tau_h=$ 18 h, $\tau_{hn} = $ 72 h and localized initial data of the form $u_n(x,0) = 0$, $u_h(x,0) = \left(1 - 800 x^2\right)_+$, $x$ being measured in mm. (a,b) Pseudocolor plots of the amplitude of the hypoxic (a) and normoxic (b) cell densities $u_h(x,t)$ and $u_n(x,t)$ respectively. The lines shown  are calculated using the propagation speed of the hypoxic and normoxic fronts in the framework of the FK model (c) Final profiles of the normoxic and hypoxic cell densities after $t=30$ days showing the coupled front and bright soliton. A small amplitude pulse of hypoxic cells leads the propagation of the tumor despite the instability of the zero solution. \label{ultima}}
\end{center}
\end{figure}

Examples of typical results are summarized in Fig. \ref{ultima}. As it can be seen in the pseudocolor plot of Fig. \ref{ultima}(a) and in Fig. \ref{ultima}(c), the hypoxic cell density does not vanish with time since their fast proliferation rates makes the zero solution unstable when the hypoxic wave hits the normal tissue areas. However, the hypoxic solution is unstable  and decays into the normoxic one leading to a bright soliton of hypoxic cells (cf. Fig. \ref{ultima}(c)) coupled to the front corresponding to the normoxic phenotype (Fig. \ref{ultima}(b)). The speed of the resulting wave of hypoxic cells becomes asymptotically close to $c^*_h$. From the mathematical point of view these solutions correspond to stable homoclinic orbits of the system connecting with the equilibrium point $u_h = 0$.

These families of solutions may correspond to the behavior of very aggressive phenotypes driven by hypoxia with both enhanced mobility and proliferation. In general, it is found in the context of different families of models that more aggressive phenotypes tend to be the drivers of invasion \citep{Anderson1,Anderson2}. What it is interesting in the result described here is the fact that even an unstable phenotype such as the hypoxic one in a well oxygenated environment may lead to a robust non-vanishing wave.

\section{Theoretical results}\label{Theoretical}
\par

In the previous sections we have described several phenomena involving the existence of both decaying and localized travelling waves of hypoxic cells. In this section we wish to complement our numerical results with some theory on the existence of localized finite-amplitude waves of the system studied. 


Let $(u_n,u_h)$ be a solution of \eqref{themodel}. Under some assumptions on the parameters to be precised later, Eqs. \eqref{themodel}  admit  {\it traveling wave}  solutions of the form $u_n(x,t)=v_n(x-ct), u_h(x,t)=v_h(x-ct),$
for some values of $c>0.$

Let $\xi:=x-ct$. The functions $(v_n,v_h)=(v_n(\xi),v_h(\xi))$ satisfy the equations
\begin{subequations}\label{eq:ode1}
\begin{eqnarray}
    -D_n v_n''-cv_n' & =(1-v_n-v_h)v_n/\tau_n+ v_h/\tau_{hn},&\\
    -D_h v_h''-cv_h' & =(1-v_n-v_h)v_h/\tau_h-v_h/\tau_{hn},&\qq{for} \xi\in\R.
\end{eqnarray}
\end{subequations}
where prime denotes differentiation with respect to $\xi.$

We will look for solutions $\big(v_n^*,v_h^*\big)$ such that
\begin{subequations}\label{bc:+inf}
\begin{eqnarray}
\big(v_n^*,(v_n^*)',v_h^*,(v_h^*)'\big)\to (0,0,0,0) & \qq{as}\xi\to +\infty \\
\label{bc:-inf}
\big(v_n^*,(v_n^*)',v_h^*,(v_h^*)'\big)\to (1,0,0,0) & \qq{as}\xi\to -\infty,
\end{eqnarray}
\end{subequations}
with $v_n^*,v_h^*>0.$

Obviously,  choosing $(v_n,v_h)=(v_K,0)$  where $v_K$ denotes the classical  KPP solution to the Fisher problem, see \citep{Fisher, KPP}, we get a solution of \eqref{eq:ode1}.
It is well known that
whenever $c>c_{n}^*$ (or $c<-c_{n}^*$), see \eqref{def:cn*} for a definition of $c_{n}^*$,  there exists a travelling wave $v_{K},$ solving
\begin{equation}\label{eq:KPP}
-D_n v_{K}''-cv_{K}'=(1-v_{K})v_{K}/\tau_n,\qq{for} \xi\in\R,
\end{equation}
and satisfying $0<v_K<1,$
\begin{eqnarray}\label{bc:vK:+-inf}
\begin{cases}
v_K(\xi)\to 0\ \text{  as }\xi\to +\infty,\\
v_K(\xi)\to 1\ \text{  as }\xi\to -\infty,\\
v'_K(\xi)\to 0\ \text{  as }|\xi|\to \infty.
\end{cases}
\end{eqnarray}
This solution corresponds to the heteroclinic orbit of the associated first order ordinary differential system,  connecting the critical points $(1,0)$ with $(0,0).$

Here we will consider solutions with $v_h>0$ and analyze the associated first order ordinary differential system, see \eqref{eq:ode3}, and its critical points. We refer to the Appendix \ref{ap1} for an study of the ODE,  the equilibria and their corresponding stability, see  Theorem \ref{th:eq:00}, and Theorem \ref{th:eq:10} where we prove that $(0,0,0,0)$ is an stable equilibrium whenever $\tau_{hn}>\tau_{h}$, and $(1,0,0,0)$ is a saddle point with a 2-dimensional unstable manifold $E^u$ (and also a 2-dimensional stable manifold $E^s$).


\begin{thm}\label{thm:heterocl} Let us denote
by $V=(V_1,V_2,V_3,V_4)=(v_n,v_n',v_h,v_h')$,
$\ V=V(\xi),$ a solution of the  ordinary differential equation \eqref{eq:ode3}.
Assume that $\tau_{hn}>\tau_{h} $ and $c>\max\{c_n^*,c_{hn}^*\},$ where 
\begin{equation}\label{def:chn*}
c_{hn}^*:=2\sqrt{D_h\left(1/\tau_{h}-1/\tau_{hn}\right)}.
\end{equation} 
 Let  $E^u$  be the 2-dimensional unstable manifold of the equilibrium $(1,0,0,0)$ characterized in Theorem \ref{th:eq:10}. Then, for any $V_0\in E^u$ there is a trajectory $\ V=V(\xi;V_0)\to (0,0,0,0)$ as $\xi\to +\infty,$ and $\ V=V(\xi;V_0)\to (1,0,0,0)$ as $\xi\to -\infty.$
\end{thm}

The proof will be presented later, in p. \pageref{proof:thm:heterocl}. Using this result we can move to the following Theorem, that  proves the existence of a 2-dimensional manifold of initial data of heteroclinic trajectories  of the associated ODE \eqref{eq:ode3},  connecting their corresponding critical points $(1,0,0,0)$ with $(0,0,0,0).$ This 2-dimensional manifold of initial data of heteroclinic trajectories, intersected with $[0,1]\times\R\times[0,1]\times\R$, is  a  set of initial data of travelling waves. The idea of the proof is to choose a point on the unstable 2-dimensional manifold $E^u$ of the equilibrium $(1,0,0,0)$ and prove that, for some  range of speeds,  the trajectory falls down into the basin of attraction of $(0,0,0,0)$.

\begin{thm}\label{thm:existence}
Assume that $\tau_{hn}>\tau_{h}$ and $c>\max\{c_n^*,c_{hn}^*\},$ where $c_n^*,c_{hn}^*$ are defined in \eqref{def:cn*}, \eqref{def:chn*} respectively.
Let  $E^u$  be the real 2-dimensional unstable manifold of the equilibrium $(1,0,0,0)$ characterized in Theorem \ref{th:eq:10}.

\smallskip

Then, for any $V_0\in E^u \cap [0,1]\times\R\times[0,1]\times\R\neq\emptyset,$ there exists a travelling wave
$\big(v_n^*, v_h^*\big)$  solving \eqref{eq:ode1}, satisfying \eqref{bc:+inf}, and such that 
$$\big(v_n^*,(v_n^*)', v_h^*,(v_h^*)'\big) \left|_{\xi=\xi_0}\right.=V_0\qq{for some} \xi_0\in\R.$$
\end{thm}

\begin{proof}
Let $(v_n,v_h)$ be a solution of \eqref{eq:ode1} and $\ V(\xi) =(V_1,V_2,V_3,V_4)=(v_n,v_n',v_h,v_h')$, satisfying Eq.   \eqref{eq:ode3}.

From  Theorem \ref{thm:heterocl}, there is a  real 2-dimensional  manifold of the equilibrium $E^u$ such that, for any initial data $V_0\in E^u,$  the trajectory $V=V(\xi;V_0) \to (1,0,0,0)$ as $\xi\to-\infty,$
and
$V=V(\xi;V_0)\to (0,0,0,0)$ as $\xi\to\infty.$

Obviously $(v_K,v_K',0,0)\left|_{\xi_1}\right. \in E^u \cap [0,1]\times\R\times[0,1]\times\R$ and $v_K(\xi_1)<1$.
Let us consider $V_0\in E^u \cap [0,1]\times\R\times[0,1]\times\R\neq\emptyset,$ then $(v_n,v_h)$ satisfies the required conditions.
\end{proof}

Next, we plan to analyze the bifurcation from $(v_n,v_h)=(v_K,0)$ to get conditions for having $v_n,v_h>0$. Let us denote by $z = (z_n,z_h):=(v_n-v_K,v_h),$ then $z$ satisfies the equations
\begin{subequations}\label{eq:ode2}
\begin{eqnarray}
     L_1z_n & = & (2-z_n-z_h)z_n/\tau_n +\left(1/\tau_{hn}
     -1/\tau_n v_K\right)z_h, \\ \label{bb}
     L_2z_h & = & (2-z_n-z_h)z_h/\tau_h-z_h/\tau_{hn},\qq{for } \xi\in\R,
\end{eqnarray}
\end{subequations}
where, 
\begin{subequations}\label{def:L1:L2}
\begin{eqnarray}
L_1 & := & -D_n\dfrac{d^2}{d\xi^2}-c\dfrac{d}{d\xi} +\dfrac{2}{\tau_{n}}v_K+\dfrac{1}{\tau_{n}}, \\
L_2 & := & -D_h\dfrac{d^2}{d\xi^2}-c\dfrac{d}{d\xi} +\dfrac{1}{\tau_{h}}v_K+\dfrac{1}{\tau_{h}},
\end{eqnarray}
\end{subequations}
in order to avoid that $0\in\s(L_1)$ or $0\in\s(L_2)$, we add $z_n/\tau_n$ and $z_h/\tau_h$ respectively to both sides, consequently $\s(L_1)\subset[1/\tau_{n},\infty)$  and $\in\s(L_2)\subset[1/\tau_{h},\infty)$ respectively, see \cite{Arendt-Batty}.

We look for solutions $(z_n,z_h)$ such that
$\big(z_n,(z_n)',z_h,(z_h)'\big)\to (0,0,0,0)$ with  $z_h>0$, as $|\xi|\to +\infty$.
Let $(z_n,z_h)$ be a solution of \eqref{eq:ode2} and let us define $Z=(Z_1,Z_2,Z_3,Z_4)=(z_n,z_n',z_h,z_h')$, then
 $\ Z=Z(\xi),$ satisfies the 
 equations
\begin{subequations}
\begin{equation}\label{eq:ode4}
\dfrac{dZ}{d\xi} = AZ+B(\xi)Z+f(\xi,Z),
\end{equation}
where
\begin{eqnarray}\label{def:A:F:z}
A  :=
\begin{pmatrix} 0 & 1 & 0 & 0 \\
 -1/(D_n \tau_{n}) & \, -c/D_n  &
  -1/(D_n\, \tau_{hn})& 0\\
0 & 0 & 0 & 1\\
0 & 0 & \, -\dfrac{1}{D_h}\left(\dfrac{1}{\tau_{h}}-\dfrac1{\tau_{hn}}\right)\, & -\dfrac{c}{D_h}
\end{pmatrix},\qquad\qquad\qquad \\
B(\xi)  :=
\begin{pmatrix} 0 & 0 & 0 & 0 \\
 2\dfrac{v_K}{D_n\, \tau_{n}} & 0 &
 \, \dfrac{v_K}{D_n\, \tau_{n}}\, & 0\\
0 & 0 & 0 & 0\\
0 & 0 &  \dfrac{v_K}{D_h\, \tau_{h}}
& 0
\end{pmatrix}, \
f(\xi,Z)  :=
\begin{pmatrix} 0 \\
 \dfrac{1}{D_n\, \tau_{n}}\, \left(Z_1 + Z_3 \right) Z_1  \\
0 \\
  \dfrac{1}{D_h\, \tau_{h}}\, \left(Z_1 + Z_3\right) Z_3
\end{pmatrix}.
\end{eqnarray}
\end{subequations}
Using \eqref{bc:vK:+-inf}, we get $
\lim_{\xi\to+\infty}B(\xi)  =:  B(+\infty) =
\left(\begin{smallmatrix} 0 & 0 & 0 & 0 \\
 0 & 0 & 0 & 0\\
0 & 0 & 0 & 0\\
0 & 0 & 0 & 0
\end{smallmatrix}\right)$, and
\begin{equation}
\lim_{\xi\to-\infty}B(\xi)  =:  B(-\infty) =
\begin{pmatrix} 0 & 0 & 0 & 0 \\
 \dfrac{2}{D_n\, \tau_{n} }& 0 &
   \, \dfrac{1}{D_n\, \tau_{n} }\, & 0\\
0 & 0 & 0 & 0\\
0 & 0 &  \dfrac{1}{D_h\, \tau_{h} }
& 0
\end{pmatrix}.
\end{equation}

The following proposition states the asymptotic behavior of a solution of \eqref{eq:ode4}, by using classical perturbation theory, see \citet[Chapter 13]{Coddington-Levinson}.
We are going to see that the minimal speed wave of the hypoxic wave is given by $c_{hn}^*,$ see \eqref{def:chn*} for a definition.

\begin{prop}\label{prop:asymptotic}
Let $Z=(Z_1,Z_2,Z_3,Z_4)$, $\ Z=Z(\xi),$ be a solution of
\eqref{eq:ode4}. Assume that $\tau_{hn}>\tau_{h}.$ Assume also that $c>\max\{c_n^*,c_{hn}^*\},$ where $c_n^*,c_{hn}^*$ are defined in \eqref{def:cn*}, \eqref{def:chn*} respectively.
Then
\begin{itemize}
\item[(i)] $\s(A):=\{\la_{1\pm},\la_{2\pm}\}\subset\R,$  where $\la_{1\pm},\la_{2\pm}$ are defined in \eqref{def:eig:z}, $\mu_{+\infty}:=\max \s(A)<0,$ and the trivial solution is
asymptotically stable. Moreover, if $|Z(0)=(z_n(0),z_n'(0),z_h(0),z_h'(0))|$
is sufficiently small, then
$$\limsup_{\xi\to+\infty}\dfrac{\log |Z(\xi)|}{\xi}\leq \mu_{+\infty}<0,
$$
where $Z(\xi)=(z_n(\xi),z_n'(\xi),z_h(\xi),z_h'(\xi)).$
\item[(ii)] $\s(A+B(-\infty)):=\{\widehat{\la}_{1\pm},\widehat{\la}_{2\pm}\}
\subset\R,$ where $\widehat{\la}_{1\pm},\widehat{\la}_{2\pm}$ are defined in \eqref{def:eig:z2}, and their eigenvalues
satisfy the following inequalities $\widehat{\la}_{1-}<0<
\widehat{\la}_{1+},$ and $ \widehat{\la}_{2-}<0<\widehat{\la}_{2+}.$

Moreover,  there exists a real
two-dimensional
manifold $S$ containing the origin, such that any solution $Z
$  of
\eqref{eq:ode4} with $Z(\xi_0)\in S
$ for any $\xi_0,$ satisfy $Z(\xi)\to 0$ as $\xi\to-\infty,$ and
$$\liminf_{\xi\to-\infty}\dfrac{\log |Z(\xi)|}{\xi}\geq
\mu_{-\infty}:=\max \{\widehat{\la}_{1-},\widehat{\la}_{2-}\}>0,
$$
where $Z(\xi)=(z_n(\xi),z_n'(\xi),z_h(\xi),z_h'(\xi))$.

\item[(iii)] Moreover, there exists a $C>0$ such that
any solution $Z$ near the origin, but not on $S$ at $\xi=\xi_0$ can not
satisfy $|Z(\xi)|\le C$ for $\xi\le \xi_0.$
\end{itemize}
\end{prop}

\begin{proof} This is a standard result in the theory of asymptotic behavior
of ordinary differential equations \citep[Chapter 13]{Coddington-Levinson}.

We have to analyze the sign of $Re(\s(A))$ where
$
\s(A):=\{\la_{1+},\la_{1-},\la_{2+},\la_{2-}\},
$
and $\la_{1\pm},\la_{2\pm}$ are defined as
\begin{equation}\label{def:eig:z}
\la_{1\pm} := \dfrac{-c\pm \sqrt{c^2-4\dfrac{D_n}{\tau_{n}}}}{2D_n},
  \la_{2\pm} := \dfrac{-c\pm\sqrt{c^2-4D_h\left(\dfrac{1}{\tau_{h}}- \dfrac{1}{\tau_{hn}}\right)}}{2D_h}.
\end{equation}
Obviously, if $\tau_{hn}\ge\tau_{h},$   then $c_{hn}^*\in\R$. Moreover, if $ c> c_{n}^*>0,$
then the eigenvalues $\la_{1\pm}$ are negative real numbers, and if $c> c_{hn}^*,$ then the eigenvalues $\la_{2\pm}$ are negative real numbers.

The eigenvalues of $\s(A+B(-\infty))$ are given by
\begin{equation}\label{def:eig:z2}
\widehat{\la}_{1\pm} := \dfrac{-c\pm \sqrt{c^2+4D_n/\tau_{n}}}{2D_n},\qquad \widehat{\la}_{2\pm} := \dfrac{-c\pm\sqrt{c^2+4 D_h/\tau_{hn}}}{2D_h}.
\end{equation}

Whenever $0<\tau_{n}<\infty,$ and $\tau_{hn}>0,$ the eigenvalues
$\widehat{\la}_{1\pm},\ \widehat{\la}_{2\pm}$ are  real numbers and
satisfy $\widehat{\la}_{1-}<0<\widehat{\la}_{1+},$ and $ \widehat{\la}_{2-}<0<\widehat{\la}_{2+}.$

\end{proof}

\bigskip

\begin{proof}{\it of Theorem \ref{thm:heterocl}.}\label{proof:thm:heterocl}
Let $\ V=V(\xi;V_0)$ be  a solution of the  Eq. \eqref{eq:ode3}.
From  Theorem \ref{th:eq:10}, there exists a  real 2-dimensional unstable manifold of the equilibrium $(1,0,0,0)$ denoted by  $E^u$. In other words, for any $V_0\in E^u,$  the trajectory $V=V(\xi;V_0)$ satisfies
$V=V(\xi;V_0)\to (1,0,0,0)$ as $\xi\to-\infty.$

\medskip

Let $z = (z_n,z_h)$ be a solution of Eq. \eqref{eq:ode2}. Since $(z_n,z_h):=(v_n-v_K,v_h),$ then $Z=(Z_1,Z_2,Z_3,Z_4)=(z_n,z_n',z_h,z_h')$,   satisfies Eq. \eqref{eq:ode4}, and  $Z=(v_n-v_K,v_n'-v_K',v_h,v_h')= V-(v_K,v_K',0,0)$.

Also, from \eqref{bc:vK:+-inf}, for any $\e>0$ there exists an $\xi_1\in\R$ such that
$$|Z(\xi)|\leq |V(\xi)-(1,0,0,0)|+|(v_K(\xi),v_K'(\xi),0,0)-(1,0,0,0)|\le \e,$$
for any $\xi\le\xi_1$. Let $Z^1=Z(\xi_1)$ and consider the trajectory $Z=Z(\xi;Z^1)$. From Proposition \ref{prop:asymptotic}, $Z=Z(\xi;Z^1)$ satisfies
$|Z|\to 0$ as $\xi\to\infty.$ Therefore, Eq.  \eqref{bc:vK:+-inf} guarantees that
$V=Z+(v_K,v_K',0,0) \to (0,0,0,0)$ as $\xi\to\infty$.
\end{proof}

\bigskip

Next, we try to find necessary conditions for having a bifurcation
from the trivial solution. Let us first pose the abstract framework
of the problem.

A weak solution of
\begin{equation}\label{eq:L:f}
L_1u_1=f_1,\qquad L_2u_2=f_2,
\end{equation}
can be defined as usual: multiplying by a test function $\psi=(\psi_1,\psi_2)\in H^1(\R)^2$, and integrating on $\R,$ we obtain $\forall \psi=(\psi_1,\psi_2)\in H^1(\R)^2,$
\begin{multline}
\int_\R D_n u_1'\psi_1'-cu_1'\psi_1+\dfrac{2}{\tau_{n}}v_K u_1\psi_1+\dfrac{1}{\tau_{n}} u_1\psi_1 \\ 
+\int_\R D_h u_2'\psi_2'-cu_2'\psi_2+\dfrac{1}{\tau_{h}}v_K u_2\psi_2+\dfrac{1}{\tau_{h}} u_2\psi_2 
 =\int_\R f_1\psi_1+f_2\psi_2.
\end{multline}

Let us set $H:=H^1(\R)^2$ with $\| u\|_H:=\big(\| u_1\|^2+\| u_2\|^2\big)^{1/2}$ and define the  bilinear form $a:H\times H\to \R$ given by
\begin{multline}
a(u,\psi):=\int_\R \left(D_n u_1'\psi_1'-cu_1'\psi_1+\dfrac{2}{\tau_{n}}v_K u_1\psi_1+\dfrac{1}{\tau_{n}} u_1\psi_1 \right)\\ 
+\int_\R \left(D_h u_2'\psi_2'-cu_2'\psi_2+\dfrac{1}{\tau_{h}}v_K u_2\psi_2+\dfrac{1}{\tau_{h}} u_2\psi_2\right),
\end{multline}
for all $u=(u_1,u_2),\ \psi=(\psi_1,\psi_2)\in H.$
\begin{thm}\label{lem:a}
Assume that $\min\left\{D_n,D_h\right\}>0,$ and that $0<\tau_{n},\tau_{h}<\infty.$
Then, for all $f=(f_1,f_2)\in H^*,$ the dual space of all linear continuous functional on $H$, there exists a unique $u=(u_1,u_2)\in H$ such that
$$a(u,\psi)=\langle f,\psi\rangle_{H^*,H} \qq{for all} \psi=(\psi_1,\psi_2)\in H,
$$
and
$$\|u\|_H\leq C\|f\|_{H^*}.
$$
The constant $C$  can be chosen as $C:= 1 /\min \left\{D_n,1/\tau_{n},D_h,1/\tau_{h}\right\}.$

\smallskip

Moreover, if $f=(f_1,f_2)\in L^2(\R)^2,$ then  $u=(u_1,u_2)\in H^2(\R)^2$,  and $\big(u_1,u'_1,u_2,u'_2\big)\to (0,0,0,0)$ as $|\xi|\to +\infty.$ 

Furthermore, if $f=(f_1,f_2)\in L^2(\R)^2\cap C(\R)^2,$ then  $u=(u_1,u_2)\in C^2(\R)^2\cap H^2(\R)^2$,
and $ \|u\|_{H^2(\R)^2}\leq C\|f\|_{L^2(\R)^2}$.
\end{thm}

\begin{proof}
The  bilinear form $a$ is continuous and coercive. Indeed, by definition
\begin{multline}
a(u,u)=\int_\R \left(D_n (u_1')^2-cu_1'u_1+\dfrac{2}{\tau_{n}}v_K u_1^2+\dfrac{1}{\tau_{n}} u_1^2\right)\\
+\int_\R \left(D_h (u_2')^2-cu_2'u_2+\dfrac{1}{\tau_{h}}v_K u_2^2+\dfrac{1}{\tau_{h}} u_2^2\right),
\end{multline}
and since
$\int_R^R u_1'u_1=\left.\frac12 u_1^2\right|_{-R}^R=\frac12 \left[u_1(R)^2-u_1(-R)^2\right]\to 0$ as $R\to\infty
$ see \cite[Corollary 8.9]{Brezis},
we obtain
\begin{eqnarray}
a(u,u) & = & \int_\R \Big(D_n (u_1')^2+\dfrac{2}{\tau_{n}}v_K u_1^2+\dfrac{1}{\tau_{n}} u_1^2
+D_h (u_2')^2+\dfrac{1}{\tau_{h}}v_K u_2^2+\dfrac{1}{\tau_{h}} u_2^2\Big) \nonumber \\
 & \geq & \min \left\{D_n,\dfrac{1}{\tau_{n}}\right\}\| u_1\|^2+\min \left\{D_h,\dfrac{1}{\tau_{h}}\right\}\| u_2\|^2 \nonumber \\
& \geq & \al \big(\| u_1\|^2+\| u_2\|^2\big)=\al \| u\|_H^2
\end{eqnarray}
for all $u=(u_1,u_2)\in H,$  for $\al:=\min\left\{D_n,1/\tau_{n}, D_h,1/\tau_{h}
\right\}>0.$  It follows that $ \al \| u\|_H^2\leq a(u,u)=\langle f,u\rangle_{H^*,H}
\leq \| f\|_{H^*}\|\| u\|_H$ and thus $  \| u\|_H\leq C\| f\|_{H^*}$ with a constant $C$ dependent only on $\al.$
The Lax-Milgram lemma completes this part of the proof.

\medskip

By hypothesis $u\in H= H^1(\R)^2.$ First, note that  if $f,u'\in L^2(\R)^2,$ then by definition of $H^1,$ $u'\in H^1(\R)^2,$
therefore $u\in  H^2(\R)^2.$
Moreover, by \citet[Corollary 8.9]{Brezis}, $u_1,u_2,u'_1,u'_2\to 0$ as $|\xi|\to +\infty.$

Due to $u\in H^2(\R)^2,$ then $u,u'\in C(\R)^2.$
Furthermore, if $f=(f_1,f_2)\in L^2(\R)^2\cap C(\R)^2,$ since $u,u',v_K,$ and $f$ are continuous on $\R$ we get $u''\in C(\R)^2$, and therefore $u\in C^2(\R)^2.$

\smallskip

Finally, taking into account that $u$ is a classical solution,
we get
\begin{equation}
\|u''\|_{L^2(\R)^2}\leq C\left(\|u\|_{H^1(\R)^2}
+\|f\|_{L^2(\R)^2}\right)\leq C\|f\|_{L^2(\R)^2},
\end{equation}
which completes the proof.
\end{proof}

\begin{cor}\label{cor:a}
Assume that $f_j\to f$ in $ L^2(\R)^2.$ Let $u_j$, $j\ge 0,$ be such that
$a(u_j,\psi)=\langle f_j,\psi\rangle_{H^*,H}$,  for all $\psi=(\psi_1,\psi_2)\in H,
$. 
Then  $u_j\to u$ in $H^2(\R)^2$
\end{cor}
\begin{proof}
It follows from theorem \ref{lem:a} taking $g_j=f_j-f$ and $v_j=u_j-u.$
\end{proof}

Let $f=(f_1,f_2),$ $f=f(z)$,  denote the nonlinearity of \eqref{eq:ode2}, i.e.
\begin{subequations}
\begin{eqnarray}
f_1(z_n,z_h) & := & \dfrac{1}{\tau_{n}}(2-z_n-z_h)z_n
+\left(\dfrac{1}{\tau_{hn}}-\dfrac{1}{\tau_{n}}v_K\right)z_h,\\
f_2(z_n,z_h) & := & \dfrac{1}{\tau_{h}}(2-z_n-z_h)z_h-\dfrac{1}{\tau_{hn}}z_h.
\end{eqnarray}
\end{subequations}

We will say that $z=(z_n,z_h)\in H$ is a {\it weak solution} of \eqref{eq:ode2} if and only if \eqref{eq:ode2} is satisfied in a weak sense, which means that
\begin{equation}
a(z,\psi)=\int_\R f_1(z_n,z_h)\psi_1+ f_2(z_n,z_h)\psi_2,\end{equation}
for all $\psi=(\psi_1,\psi_2)\in H.$

The following result ensures that a weak solution is a classical solution, and provide the rates at $\pm\infty.$

\begin{thm}\label{lem:nol}
Assume that
$\min\left\{D_n,D_h\right\}>0,$ and that $0<\tau_{n},\tau_{h}<\infty.$
Assume that $\tau_{hn}>\tau_{h}$ and $c>\max\{c_n^*,c_{hn}^*\},$ where $c_n^*,c_{hn}^*$ are defined in \eqref{def:cn*}, \eqref{def:chn*} respectively. Let $z=(z_n,z_h)\in H$ be a weak solution of \eqref{eq:ode2}.

Then the following holds:
\begin{enumerate}
\item[i)] $z=(z_n,z_h)\in C^2(\R)^2\cap H^2(\R)^2,$ and
$
\|z\|_{H^2(\R)^2}\le C\|z\|_{L^2(\R)^2}.$

\item[ii)] Let $\ Z=(z_n,z_n',z_h,z_h')$, then
$$\limsup_{\xi\to+\infty}\dfrac{\log |Z(\xi)|}{\xi}\leq \mu_{+\infty}<0,
\qq{and}
\liminf_{\xi\to-\infty}\dfrac{\log |Z(\xi)|}{\xi}\geq
\mu_{-\infty}>0,
$$
where $\mu_{+\infty}:=\max \s(A) =\{\la_{1+}, \la_{1-},\la_{2+},\la_{2-}\}
$,
$\mu_{-\infty}:=\max \{\widehat{\la}_{1-},\widehat{\la}_{2-}\},$ and $\la_{1\pm},$ $\la_{2\pm},\widehat{\la}_{1-},\widehat{\la}_{2-}$ are defined in \eqref{def:eig:z}, \eqref{def:eig:z2} respectively.
\end{enumerate}
\end{thm}
\begin{proof}
i) By Morrey's Theorem, see \cite{Brezis},  there is a constant
$C$  such that
$$|z_i(x)-z_i(y)|\le C\|\nabla z_i\|_{L^2(\R)}|x-y|^{1/2},\qq{a.e.} x,y\in\R, \qquad i\in\{n,h\}.
$$
Moreover, due to $z\in H$, we get that $\big(z_n,z_h\big)\to (0,0)$ as $|\xi|\to +\infty$.
Therefore, applying Theorem \ref{lem:a}, we complete this part of the the proof.

ii) From part i) we get that $\big(z_n,z'_n,z_h,z'_h\big)\to (0,0,0,0)$ as $|\xi|\to +\infty$
and therefore, applying Proposition \ref{prop:asymptotic} we conclude the proof.
\end{proof}

\begin{lem}\label{lem:bdd}
Assume that $\tau_{hn}> \tau_{h}$ and $c>\max\{c_n^*,c_{hn}^*\},$ where $c_n^*,c_{hn}^*$ are defined in \eqref{def:cn*}, \eqref{def:chn*} respectively.
Let $(v_n,v_h)$ be a nonegative solution of \eqref{eq:ode1}, with $v_n\not\equiv 0$, $v_h\not\equiv 0$. Then, 
\begin{equation}\label{def:Mv}
0< v_n\leq \max\left\{1,\tau_{n}/\tau_{hn}\right\},
\qq{and}0< v_h\leq 1- \tau_h/\tau_{hn}.
\end{equation}

\end{lem}
\begin{rk}\label{rk:bdd:nolH2}
Since $(v_n,v_h)$ is a nonnegative solution of Eq. \eqref{eq:ode1}, then $(z_n,z_h)=(v_n-v_K,v_h)$ is a  solution of Eq. \eqref{eq:ode2}. Therefore, Lemma \ref{lem:bdd} implies that
$$
-1\leq z_n\leq \max\left\{1,\tau_{n}/\tau_{hn}\right\},\qquad 0\leq z_h\leq 1- \tau_h/\tau_{hn}.
$$

\end{rk}

\begin{proof}{\it of Lemma \ref{lem:bdd}.}
Let $(v_n,v_h)=(z_n+v_K,z_h)$ be a nonnegative solution of \eqref{eq:ode1} with $v_n\neq 0$, $v_h\neq 0$. Note that $v_n,v_h$ can only have simple zeros, in the sense that if $v_n(x)=0,$ then $v_n'(x)\neq 0$. Therefore $v_n,v_h>0$ on $\R.$

Let $x_n,x_h\in\R$ be such that $v_n(x_n)=\max_\R{v_n}>0$, $v_h(x_h)=\max_\R{v_h}>0.$ Then, $v'_i(x_i)=0,\ v''_i(x_i)\leq 0,$ $i\in \{n,h\}$ and by \eqref{eq:ode1} we obtain
$$
-D_hv_h''(x_h)=\dfrac{1}{\tau_{h}}(1-v_n(x_h)-v_h(x_h))v_h(x_h)
-\dfrac{1}{\tau_{hn}}v_h(x_h)\geq 0.
$$
Dividing by $v_h(x_h)>0,$ we get
$
1/\tau_{h}-1/\tau_{hn}
\geq \left(v_n(x_h)+v_h(x_h)\right)/\tau_h\geq v_h(x_h)/\tau_h
$
and the  inequalities related to $v_h$ in \eqref{def:Mv} are attained.

Likewise
$
-D_nv_n''(x_n)=\dfrac{1}{\tau_{n}}(1-v_n(x_n)-v_h(x_n))v_n(x_n)
+\dfrac{1}{\tau_{hn}}v_h(x_n)\geq 0,
$
and therefore
$
\dfrac{1}{\tau_{n}}\left(1-v_n(x_n)\right)v_n(x_n)
+\left(\dfrac{1}{\tau_{hn}}-\dfrac{1}{\tau_{n}}v_n(x_n)\right)v_h(x_n)\geq 0.
$

If $v_n(x_n)\geq \tau_{n}/\tau_{hn}$ then
$$
(1-v_n(x_n))v_n(x_n)/\tau_n \geq
\left(v_n(x_n)/\tau_n-1/\tau_{hn}\right)v_h(x_n)\geq 0,
$$
and dividing by $v_n(x_n)>0,$ we get
$v_n(x_n)\leq 1.$ Also $v_n(x_n)\leq \tau_{n}/\tau_{hn},$
and the  inequalities related to $v_n$ in \eqref{def:Mv} are attained, completing the proof.
\end{proof}

Finally, we move to one of our main results in this section that states necessary conditions for having  solutions bifurcating from the solution $(v_K,0)$ and positive in the second component, i.e. persistent finite amplitude solitons composed of cells with an hypoxic phenotype. Technically we use classical bifurcation methods. The convergence in any bounded interval is clear, and also that the limit function is a classical solution of the limit problem in all $\R.$ The difficulty comes when trying to prove that the convergence is, in fact, in all $\R$ and as a consequence, that the limit function is non-trivial.

\smallskip

To achieve it, we need uniform estimates of the $H^1$ norms in the exterior of bounded intervals. In particular, we have to control the behavior at $\pm\infty.$ We overcome this difficulty by using perturbation theory in this framework.

\medskip

The interesting asymptotic behavior is obtained when 
$\tau_{hn}>\tau_{h}$, and $c>\max\{c_n^*,c_{hn}^*\}$,
where $c_n^*,c_{hn}^*$ are defined in \eqref{def:cn*} and \eqref{def:chn*} respectively, see Appendix and Proposition \ref{prop:asymptotic}.

\begin{thm}\label{lem:cn:bif}

Let $(\La_j,z_j)\in \R^3\times H,$  be a sequence solving
\eqref{eq:ode2}, with $\La_j:=(\tau_{n,j},\tau_{h,j},\tau_{hn,j}) \in \R^3,$ and $z_j=(z_{n,j},z_{h,j})\in H$ with $z_{h,j}> 0$. Assume  that $\|z_j\|_{L^2(\R)^2}\to 0.$ Assume also that
$
\limsup_{j\to\infty}\tau_{h,j}=:\overline{\tau}_{h}>0,
$
that
$\lim_{j\to\infty}\tau_{hn,j}=:\overline{\tau}_{hn}>0,$ 
that
$\tau_{hn,j}>\tau_{h,j},$
and that
$ c>\max\{c_{n,j}^*,c_{hn,j}^*\},$
for some subsequence again denoted by $\tau_{hn,j}.$
Then,
$$\dfrac{2}{\overline{\tau}_{h}}\, -\dfrac{1}{\overline{\tau}_{hn}}\in\s\left(L_2\right).
$$

Moreover, at least for a subsequence,
$$
\dfrac{z_{h,j}}{\|z_{h,j}\|_{L^2(\R)}}\to \Phi_h> 0 \qq{in} H^1(\R),\qquad \|\Phi_h\|_{L^2(\R)}=1,
$$
and the function
$\Phi_h \in C^2(\R)\cap H^2(\R)$ satisfies
\begin{equation}\label{eq:phi2}
-D_h\Phi_h''-c\Phi_h' +\dfrac{1}{\overline{\tau}_{h}}\,\, v_K\Phi_h
=\left(\dfrac{1}{\overline{\tau}_{h}}\,-\dfrac{1}{\overline{\tau}_{hn}}\right)\Phi_h.
\end{equation}
\end{thm}

\begin{proof} Let us choose a subsequence such that
$(\tau_{n,j},\tau_{h,j},\tau_{hn,j},\|z_j\|_{L^2(\R)^2}) \to ( \overline{\tau}_{n}, \overline{\tau}_{h}, \overline{\tau}_{hn},0).$ Moreover,
$\|z_{j}\|_{H}\le C\|z_{j}\|_{L^2(\R)^2}\to 0.$ Also, see Theorem \ref{lem:a} and Remark
\ref{rk:bdd:nolH2},
$\|z_{j}\|_{H^2(\R)^2}\le C\|z_{j}\|_{L^2(\R)^2}\to 0.$

Let $\Phi_{hj}:=z_{h,j}/\|z_{h,j}\|_{L^2(\R)}\ge 0.$ Then $\|\Phi_{hj}\|_{L^2(\R)}=1$ and
\begin{equation}\label{eq:phi:norm0}
\|\Phi_{hj}\|_{L^\infty(\R)}\le \|\Phi_{hj}\|_{H^1(\R)}\le C,\qq{and}
\|\Phi_{hj}\|_{H^2(\R)}\le C.
\end{equation}

From the Kolmogorov-Riesz-Fr\'echet's Theorem, see
\citep[Theorem 4.26 and Corollary 4.27]{Brezis}, for any bounded
interval $I$ fixed, there exists a subsequence  $\Phi_{hj}\to \Phi_h$ in $H^1(I).
$
If moreover, for any $\e>0,$ there exists a bounded interval $I=I(\e)$ such that
\begin{equation}\label{ineq:phi}
\|\Phi_{hj}\|_{H^1(\R\setminus I)}<\e,
\end{equation}
then, there exists a subsequence such that
$\Phi_{hj}\to \Phi_h$ in $H^1(\R).
$

Let us fix   $R_0>0,$ then $\Phi_{hj}\to \Phi_{h}$ in $H^1(-R_0,R_0),$
where  $\Phi_{h}$ depends on $R_0,$ (we shall denote it by $\Phi_{h,R_0}$ when we need  to remark this dependence).
In order to achieve \eqref{ineq:phi},  let us divide Eq. \eqref{bb} by $\|z_{2j}\|_{L^2(\R)}$,  then we obtain
\begin{equation}\label{eq:phi2j}
L_2\Phi_{hj}=(2-z_{n,j}-z_{h,j})\Phi_{hj}/\tau_{h,j}
-Phi_{hj}/\tau_{hn,j}\
\end{equation}

Let $\zeta_j=(\Phi_{hj},\Phi'_{hj})$ be the solution of the associated IVP
\begin{eqnarray}\label{eq:ode0}
\dfrac{d\zeta}{d\xi} &=& A\zeta+B(\xi)\zeta+(A_j-A)\zeta+ [B_j(\xi)-B(\xi)]\zeta, \nonumber\\
\zeta(0)&=& (\Phi_{hj}(0),\Phi'_{hj}(0)),
\end{eqnarray}
where
\begin{subequations}\label{def:F0}
\begin{equation*}
A_j :=
\left(\begin{array}{cc}  0 & 1\\
\dfrac{1}{D_h} \left(\dfrac{1}
{\tau_{hn,j}}-\dfrac{1}{\tau_{h,j}}
\right)
& \, - \dfrac{c}{D_h}
\end{array}\right),\ \
A :=
\left(\begin{array}{cc}  0 & 1\\
\dfrac{1}{D_h} \left(\dfrac{1}{\overline{\tau}_{hn}}-\dfrac{1}{\overline{\tau}_{h}}
\right)
& \, - \dfrac{c}{D_h}
\end{array}\right),
\end{equation*}
\begin{equation*}
B_j(\xi) :=
\left(\begin{array}{cc}  0 & 0 \\
 \dfrac{v_K}{D_h\tau_{h,j}}
+\dfrac{z_{n,j}+z_{h,j}}{D_h\, \tau_{h,j}}\,\,
& 0
\end{array}\right),\qquad
B(\xi) :=
\left(\begin{array}{cc}  0 & 0 \\
 \dfrac{v_K}{D_h\overline{\tau}_{h}}
& 0
\end{array}\right),\
\end{equation*}
\end{subequations}
and $A_j\to A$ as $j\to\infty$, $B_j(\xi)\to B(\xi)$   as $j\to\infty$ for each $\xi\in\R$. Moreover,
using \eqref{bc:vK:+-inf} and since $(z_{n,j},z_{h,j})\to (0,0)$ as $j\to\infty$, see Theorem \ref{lem:nol}, we get
\begin{subequations}
\begin{eqnarray}
\lim_{\xi\to+\infty}B_j(\xi) & =: & B_j(+\infty) =
\left(\begin{smallmatrix} 0 & 0  \\
 0 & 0
\end{smallmatrix}\right),
\\
\lim_{\xi\to-\infty}B_j(\xi) & =: & B_j(-\infty) =
\begin{pmatrix} 0 & 0  \\
1/(D_h\tau_{h,j}) & 0
\end{pmatrix}
\end{eqnarray}
\end{subequations}

Eq. \eqref{eq:ode0} is an homogeneous linear problem. From classical perturbation theory (see \cite{Coddington-Levinson}), we get the asymptotic behavior
\begin{subequations}\label{phi2:inf}
\begin{eqnarray}\label{phi2:infa}\limsup_{\xi\to+\infty}\dfrac{\log |(\Phi_{hj}(\xi),\Phi'_{hj}(\xi))|}{\xi}\leq
(\la_{2+})+\e<0,
\end{eqnarray}
and
\begin{eqnarray}\label{phi2:infb}\liminf_{\xi\to+\infty}\dfrac{\log |(\Phi_{hj}(-\xi),\Phi'_{hj}(-\xi))|}{-\xi}\geq
(\widehat{\la}_{2+})-\e>0.
\end{eqnarray}
\end{subequations}
where $\la_{2+},\widehat{\la}_{2+}$ are defined in \eqref{def:eig:z},\eqref{def:eig:z2} respectively.

\medskip

Therefore \eqref{ineq:phi} is accomplished,  $\Phi_{hj}\to \Phi_h$ in $H^1(\R)$ and $\Phi_h$ solves \eqref{eq:phi2}.
In particular
$\|\Phi_h\|_{L^2(\R)}=1$ which implies $\Phi_h\neq 0$, and then $\Phi_h\gneqq 0.$ The Maximum Principle implies that $\Phi_h> 0.$

\medskip

Multiplying \eqref{eq:phi2} by $\Phi_h$ and integrating on $\R$ we obtain
$$
D_h\int_\R (\Phi_h')^2 + \dfrac{1}{\overline{\tau}_{h}} \int_\R  v_{K}\Phi_h^2= \left(\frac{1}{\overline{\tau}_{h}}\, -\frac{1}{\overline{\tau}_{hn}}\right)\int_\R \Phi_h^2,
$$
so $1/\overline{\tau}_{h}  -1/\overline{\tau}_{hn}\ge 0$
which completes the proof.
\end{proof}

\bigskip

\begin{rk} The asymptotic behavior provided by \eqref{phi2:infa} is un upper bound. In the same sprit as Proposition \ref{prop:asymptotic}, (ii),  
there exists a real
one-dimensional
manifold $S_1$ containing the origin, such that any solution $(\Phi_{h},\Phi'_{h})$  of  
\begin{equation}\label{eq:ode10}
\dfrac{d\zeta}{d\xi} = A\zeta+B(\xi)\zeta,\qquad \zeta(0)= (\Phi_{h}(0),\Phi'_{h}(0)),
\end{equation}
where $A,B$ are defined in \eqref{def:F0}, and
with $ (\Phi_{h}(\xi_0),\Phi'_{h}(\xi_0))\in S_1
$ for any $\xi_0,$ satisfy $\Phi_{h}(\xi),\Phi'_{h}(\xi)\to 0$ as $\xi\to\infty,$ and
\begin{equation}\label{onda:la:2-}
\limsup_{\xi\to\infty}\dfrac{\log |(\Phi_{hj}(\xi),\Phi'_{hj}(\xi))|}{\xi}\leq
(\la_{2-})+\e<0,
\end{equation}
where $\la_{2-}$ is defined in \eqref{def:eig:z}.
\end{rk}

From a theoretical point of view, the hypoxic front wave can have two 'components' when $\xi\to\infty$, one corresponding to the biggest eigenvalue $\la_{2+}$, and the other one to the smallest eigenvalue $\la_{2-}$. From a quantitative point of view,  the component corresponding to the biggest eigenvalue $\la_{2+}$, is much bigger than  the component corresponding to the smallest eigenvalue $\la_{2-}$.
The above result states that the speed wave of the component of the hypoxic front wave corresponding to the smallest eigenvalue  converges to $c_{hn}^*.$

\begin{cor}\label{cor:main} 

Let $(\La_j,z_j)\in \R^4\times H,$  be a sequence solving
\eqref{eq:ode2}, with $\La_j:=(\tau_{n,j},\tau_{h,j},\tau_{hn,j},c_j) \in \R^4,$ and $z_j=(z_{n,j},z_{h,j})\in H$ with $z_{h,j}> 0$. Assume  that $\|z_j\|_{L^2(\R)^2}\to 0.$ Assume also that
$
\limsup_{j\to\infty}\tau_{h,j}=:\overline{\tau}_{h}>0,
$
that
$\lim_{j\to\infty}\tau_{hn,j}=:\overline{\tau}_{hn}>0,$
that
$\tau_{hn,j}>\tau_{h,j},$
that
$ c_j>\max\{c_{n,j}^*,c_{hn,j}^*\},$
and  that
$\lim_{j\to\infty}c_{j}=:\overline{c}>0,$ for some subsequence again denoted by $\La_j.$

Let us keep the notation of Theorem \ref{lem:cn:bif}.

\smallskip

Assume that for any $\e>0$ there exists a $j_0$ such that if $j\ge j_0$ then 
\begin{equation}\label{onda:la:2-j}
\limsup_{\xi\to\infty}\dfrac{\log |(\Phi_{hj}(\xi),\Phi'_{hj}(\xi))|}{\xi}\leq
(\la_{2-,j})+\e<0,
\end{equation}
where $\la_{2-,j}$ are defined in \eqref{def:eig:z} for $(\tau_{h},\tau_{hn},c)=(\tau_{h,j},\tau_{hn,j},c_j)$, 
then
\begin{equation}\label{def:chn*:2}
\overline{c}= c_{hn}^*= 2\sqrt{D_h\left(\frac{1}{\overline{\tau}_{h}}\, -\frac{1}{\overline{\tau}_{hn}}\right)}.
\end{equation}

\end{cor}
\begin{proof}
By hypothesis we have that
$c_j>\max\{c_{n,j}^*,c_{hn,j}^*\}   \ge c_{hn,j}^*. $ We shall argue by contradiction, assuming that $\overline{c}> c_{hn}^*.$ 
By definition of   $\la_{2-}$, see \eqref{def:eig:z}, we can state that $\la_{2-}<-\overline{c}/2D_h,$
therefore, taking into account \eqref{onda:la:2-},  there exists some $\e>0$ such that
\begin{equation}
\limsup_{\xi\to+\infty}\dfrac{\log |(\Phi_{h}(\xi),\Phi'_{h}(\xi))|}{\xi}\leq
-\frac{\overline{c}}{2D_h}-\e<0,
\end{equation}

Let $\varphi:=e^{\alpha \xi} \Phi_h$  for some $\al$ to be determined later.
Differentiating twice, substituting into the equation  (\ref{eq:phi2}), and rearranging terms we can write
\begin{equation}\label{eq:der:KPP6}
-D_h \varphi''+\left(2D_h\al-\overline{c}\right)\varphi'+\left(-D_h\al^2+\overline{c}\al+ v_{K}/\overline{\tau}_{h}\right)\varphi=\s\varphi, \text{for} \ \xi\in\R,
\end{equation}
where $\s=1/\overline{\tau}_{h}\, -1/\overline{\tau}_{hn}$. Choosing $\al:=\overline{c}/2D_h$ we get 
$$\limsup_{\xi\to+\infty}\dfrac{\log |(\varphi(\xi),\varphi'(\xi))|}{\xi}\leq -\e<0,$$ therefore $\varphi\in H^1(\R)$ and $
-D_h \varphi''+\left(\frac{\overline{c}^2}{4D_h}+\dfrac{1}{\overline{\tau}_{h}} v_{K}\right)\varphi=\s\varphi$, for $\xi\in\R.$
We can consider this as a self-adjoint problem in $L^2(\R),$ regardless of the original space, so that any eigenvalue must be real. The smallest eigenvalue is
$$\s=\inf_{\varphi\in H^1(\R):\int_\R \varphi^2=1} \left\{D_h\int_\R (\varphi')^2 + \int_\R \left(\frac{\overline{c}^2}{4D_h}+\dfrac{1}{\overline{\tau}_{h}} v_{K}\right)\varphi^2\right\}\ge \frac{\overline{c}^2}{4D_h},
$$
and therefore $\s \ge \overline{c}^2/4D_h.$ Taking into account that  $\s=1/\overline{\tau}_{h}\, -1/\overline{\tau}_{hn}$, we obtain $1/\overline{\tau}_{h}\, -1/\overline{\tau}_{hn} \ge \overline{c}^2/4D_h$ or equivalently, that $\overline{c} \le  2\sqrt{D_h\left(1/\tau_{h}-1/\tau_{hn}\right)}=: c_{hn}^* ,$ which contradicts the hypothesis, ending the proof.
\end{proof}

\bigskip

With respect to the first component, we can prove the following result.

\begin{cor} If the hypothesis of Theorem \ref{lem:cn:bif},
hold, then
$$\dfrac{z_{n,j}}{\|z_{j}\|_{L^2(\R)^2}}\to \Phi_n \qq{in} H^1(\R),
$$
where $\|z\|_{L^2(\R)^2}=\left(\|z_{n}\|_{L^2(\R)^2}^2
+\|z_{h}\|_{L^2(\R)^2}^2\right)^{1/2}.$

Moreover, 
$\Phi_n \in C^2(\R)\cap H^2(\R)$ is such that $\Phi_n\neq 0,$ and for some $\te\in[0,1)$
\begin{equation}\label{eq:phi1}
-D_n\Phi_n''-c\Phi_n' +\dfrac{2}{\overline{\tau}_{n}}v_K\Phi_n
=\dfrac{1}{\overline{\tau}_{n}} \Phi_n+\te\left(\dfrac{1}{\overline{\tau}_{hn}}
-\dfrac{1}{\overline{\tau}_{n}}v_K\right) \Phi_h.
\end{equation}

\end{cor}

\begin{proof}
Let $\Phi_{nj}:=z_{n,j}/\|z_{j}\|_{L^2(\R)^2},$ where
$\|z\|_{L^2(\R)^2}^2=\|z_{n}\|_{L^2(\R)^2}^2+\|z_{h}\|_{L^2(\R)^2}^2.$ Then $\|\Phi_{nj}\|_{L^2(\R)}\leq 1$ and
\begin{equation}\label{eq:phi:norm0:1}
\|\Phi_{nj}\|_{L^\infty(\R)}\le \|\Phi_{nj}\|_{H^1(\R)}\le C,\qq{and}
\|\Phi_{nj}\|_{H^2(\R)}\le C.
\end{equation}

From the Kolmogorov-Riesz-Fr\'echet's Theorem, see
\citet[Theorem 4.26 and Corollary 4.27]{Brezis}, for any bounded
interval $I$ fixed, there exists a subsequence
$\Phi_{nj}\to \Phi_n \text{in }H^1(I).
$ Let us fix   $R_0>0,$ then
$\Phi_{nj}\to \Phi_{n,R_0}  \text{ in }H^1(-R_0,R_0).
$

Obviously $\|z_{h,j}\|_{L^2(\R)}/\|z_{j}\|_{L^2(\R)^2}\leq 1.$
Then, at least for a subsequence,
$\frac{\|z_{h,j}\|_{L^2(\R)}}{\|z_{j}\|_{L^2(\R)^2}}\to\te\in
[0,1], \text{and therefore,}$
$\|\Phi_{n,j}\|_{L^2(\R)}=
\frac{\|z_{1j}\|_{L^2(\R)}}{\|z_{j}\|_{L^2(\R)^2}}
\to\sqrt{1-\te^2}.$

Dividing  by $\|z_{j}\|_{L^2(\R)}$ the first
equation of \eqref{eq:ode2},  we obtain
\begin{equation}\label{eq:phi1j}
L_1\Phi_{nj}=\dfrac{1}{\tau_{n,j}}(2-z_{n,j}-z_{h,j})\Phi_{nj}
+\left(\dfrac{1}{\tau_{hn,j}}-\dfrac{v_K}{\tau_{n,j}}\right)
\te\Phi_{hj}.
\end{equation}

Let now
$\zeta=(\Phi_{nj}, \Phi'_{nj})$ be the solution of the associated  IVP
\begin{equation}\label{eq:ode00}
\dfrac{d\zeta}{d\xi} =A\zeta+B(\xi)\zeta+F(\xi)+(A_j-A)\zeta+ [B_j(\xi)-B(\xi)]\zeta+[F_j(\xi)-F(\xi)],
\end{equation}
with $\zeta(0)=(\Phi_{nj}(0),
\Phi'_{nj}(0)),$ 
where
\begin{equation*}\label{def:F00}
A_j :=
\left(\begin{array}{cc}  0 & 1\\
-\dfrac{1}{D_n\, \tau_{n,j}}
& \, - \dfrac{c}{D_n}
\end{array}\right),\ \
A :=
\left(\begin{array}{cc}  0 & 1\\
-\dfrac{1}{D_n\,\overline{\tau}_{n}}
& \, - \dfrac{c}{D_n}
\end{array}\right),
\end{equation*}
\begin{equation*}\label{def:F01}
B_j(\xi) :=
\left(\begin{array}{cc}  0 & 0 \\
 \dfrac{2v_K}{D_n\tau_{n,j}}
+\dfrac{z_{n,j}+z_{h,j}}{D_n\, \tau_{n,j}}\,\,
& 0
\end{array}\right),\qquad
B(\xi) :=
\left(\begin{array}{cc}  0 & 0 \\
 \dfrac{2v_K}{D_n\overline{\tau}_{n}}
& 0
\end{array}\right),\
\end{equation*}
\begin{equation*}\label{def:F02}
F_j(\xi):=
\left(\begin{array}{c}  0  \\
\theta\left(\dfrac{v_K}{D_n\tau_{n,j}} -\dfrac{1}{\tau_{hn,j}}\right)\Phi_{hj}
\end{array}\right),\
F(\xi):=
\left(\begin{array}{c}  0  \\
\theta\left(\dfrac{v_K}{D_n\overline{\tau}_{n}} -\dfrac{1}{\overline{\tau}_{hn}}\right)\Phi_{h}
\end{array}\right),\
\end{equation*}
and $A_j\to A$ as $j\to\infty$; $B_j(\xi)\to B(\xi)$,  $F_j(\xi)\to F(\xi)$ as $j\to\infty$ for each $\xi\in\R$. Moreover,
using \eqref{bc:vK:+-inf} and the fact that $(z_{n,j},z_{h,j})\to (0,0)$ as $j\to\infty$, see Theorem \ref{lem:nol}, we get
\begin{subequations}
\begin{eqnarray}
\lim_{\xi\to+\infty}B_j(\xi) & =: & B(+\infty) =
\left(\begin{smallmatrix} 0 & 0  \\
 0 & 0
\end{smallmatrix}\right),
\\
\lim_{\xi\to-\infty}B_j(\xi) & =: & B_j(-\infty) =
\begin{pmatrix} 0 & 0  \\
2/(D_n\tau_{n,j}) & 0
\end{pmatrix}.
\end{eqnarray}
\end{subequations}
From \eqref{phi2:inf}, we get
$$\limsup_{\xi\to+\infty}\dfrac{\log |F_j(\xi)|}{\xi}\leq
(\la_{2+})+\e<0,
$$
and
$$\liminf_{\xi\to+\infty}\dfrac{\log |F_j(\xi)|}{-\xi}\geq
(\widehat{\la}_{2+})-\e>0,
$$
where $\la_{i\pm},\widehat{\la}_{2+}$,  are defined in \eqref{def:eig:z},\eqref{def:eig:z2} respectively.

Eq. \eqref{eq:ode00} is a non-homogeneous linear problem. From the variation of parameters formula, we get, as before, its asymptotic behavior
$$\limsup_{\xi\to+\infty}\dfrac{\log |(\Phi_{nj}(\xi),\Phi'_{nj}(\xi))|}{\xi}\leq
\max\{\la_{1+},\la_{2+}\}+\e<0,
$$
and
$$\liminf_{\xi\to+\infty}\dfrac{\log |(\Phi_{nj}(-\xi),\Phi'_{nj}(-\xi))|}{-\xi}\geq
\min\{\widehat{\la}_{1+},\widehat{\la}_{2+}\}-\e>0.
$$
where $\widehat{\la}_{i+}, \la_{i+}$, $i=1,2$, are defined in \eqref{def:eig:z},\eqref{def:eig:z2} respectively.
Then $\Phi_{nj}\to\Phi_n$ in $H^1(\R) $ and $\Phi_n$ solves \eqref{eq:phi1}.
If $\te=1$ then $\|\Phi_n\|_{L^2(\R)^2}
=\sqrt{1-\te^2}=0$ which  contradicts \eqref{eq:phi1}, due to $\|\Phi_h\|_{L^2(\R)^2}=1,$ so $0\le\te<1$ and $\Phi_n\neq 0$.

\end{proof}

\section{Therapeutical implications and conclusions}
\label{discussion}

Hypoxia is a characteristic feature of high-grade gliomas HGGs and arises first as a result of the proliferative activity of cells overcoming the capabilities of oxygen supply by the vasculature. In the case of gliomas there is an additional effect due to the secretion of prothrombotic factors that result in vessel failure. It is interesting to note that hypoxia is only marginal in low grade gliomas \citep{Zagzag2000}, where the vasculature remains mainly intact. Our results indicate that hypoxic events will result in an accelerated progression even when those hypoxic events are local in time, leading to a (may be diffuse and/or small in amplitude) front of invasive cells displaying the hypoxic phenotype (see e.g. Figure 3C of \citet{Zagzag2000}). This fact limits the potential efficacy of therapies targeting oxygenation alone, such as those described by \citet{Hatzikirou2012} or
\citet{Alicia2012} since it is not possible to stop completely the occurrence of hypoxic events in such an aggresive type of tumor.

Is there then an alternative to use re-oxygenation to favor the more proliferative yet less invasive phenotypes? The only possibility in the framework of the simple description used in this paper is to act on the switching time $\tau_{hn}$. In real situations $\tau_{hn}$ is not constant but typically increases with the number of hypoxia cycles experienced by cells until they reach to a state of physical balance with HIF-1$\alpha$. This sequence of  oxygen deprivation episodes drives genetic alterations in tumor cells so that HIF-1$\alpha$ is accumulated in their nucleus even in oxic conditions and thus cells can not return to their previous state \citep{Semenza2003}. Therefore, $\tau_{hn}$ becomes larger than the typical proliferation time. Once the cells take so much time to revert to their less motile state in oxic conditions, the speed of the front increases in a sustained way (see Fig. \ref{cuarta}).
In fact, in vivo analysis of HIF-1 stabilization in well oxygenated tumor areas \citep{Zagzag2000} supports long normalization times.
However, if $\tau_{hn}$ could be kept small, as it happens in the normal cellular physiological state, the tumor invasion speed would drastically decrease to $c^*_n$ in oxic conditions. Thus, in order to be effective, a therapy involving enhanced oxygenation (or vascular normalization) should also act on HIF-1$\alpha$ equilibrium simultaneously. This fact, in addition to other reasons, may be a reason for the failure of antiangiogenic therapies that lead to a radiographic response and increase in the progression free survival \citep{Friedman2009} but not to a real increase in survival of GBM patients \citep{Butowski2011}.
Moreover it is known that in recurrent tumors the effect of antiangiogenic therapy is greatly reduced and a hypothesis for the lack of response after antiangiogenic treatment is an alteration of the tumor phenotype in a highly infiltrative compartment that is angiogenic-independent  \citep{Beal2011}.

A final implication of our results affects the transition from low grade glioma (LGG) to higher grades for the GBMs. Although many factors may induce the transition, the simplest explanation involves the development of hypoxic areas \citep{Swanson2011}.
This idea matches well with the fact that several distinctive features between LGGs and HGGs are related to the presence of hypoxia in the later ones: pseudopalisades and necrosis, microvascular proliferation and a higher cellularity originating the hypoxic events \citep{vaso-occlusive,Alicia2012}. Our results show that preventing vessel failure and the cascade of malignant transformations associated to hypoxia may result in a delay in the appearance of the more invasive phenotype.
Thus LGGs suspicious to undergo the malignant transformation might benefit from  anti-thrombotic medication to avoid the degeneration into HGGs.

Because of the relevant role of hypoxia-inducible factors on the aggressiveness and tumorgenic capacity of glioma cells \citep{Li2009} stabilization of HIF-1$\alpha$ has been also recently proposed as an attractive therapeutic target \citep{Semenza2003}. Our study points out yet another reason based on purely dynamical considerations. The stabilization of HIF-1$\alpha$ would lead to smaller invasion speeds and thus to slower glioma progression if combined with therapies improving tumor oxygenation.

In summary, we have studied a simple model that shows the large impact of localized in time hypoxic events in the progression of gliomas. The mechanism is based on a low amplitude wave of hypoxic cells  that seed the normal tissue in advance and accelerate the progression of the wave of more proliferative cells. The key parameter in this process is the time
that hypoxic phenotype cells take to revert to the normoxic phenotype under oxic environmental conditions.

We have also provided some theoretical results including necessary conditions for having persistent finite amplitude solitons composed of cells with an hypoxic phenotype. From a quantitative point of view,  the hypoxic front wave has two 'components'
one much bigger than  the other one. We have proven that the speed wave of the 'smallest' component of the hypoxic front wave  converges to $c_{hn}^*.$ The stability of those wave packets is a difficult mathematical problem, due to the essential spectrum and will be studied in the future.

\begin{appendix}

\section{Critical points}
\label{ap1}
Let $(v_n,v_h)$ be a solution of \eqref{eq:ode1} and define $V(\xi) = (V_1,V_2,V_3,V_4)=(v_n,v_n',v_h,v_h')$, then \begin{equation}\label{eq:ode3}
\dfrac{dV}{d\xi} = F(V),
\end{equation}
where
\begin{equation}\label{def:F}
F(V) :=
\left(\begin{array}{c} V_2 \\ - \frac{c}{D_n}\, V_2 -\dfrac{1}{D_n\, \tau_{n} }\, \left(1-V_1 - V_3 \right)\, V_1  - \dfrac{1}{D_n\, \tau_{hn}}V_3 \\ V_4 \\ - \frac{c}{D_h}\, V_4 - \dfrac{1}{D_h\, \tau_{h} }\, \left(1-V_1 - V_3\right)\, V_3 +\dfrac{1}{D_h\, \tau_{hn}}V_3\end{array}\right).
\end{equation}

A point $\big(\bar{V}_1,\bar{V}_2,\bar{V}_3,\bar{V}_4\big)$ is an {\it equilibrium} of \eqref{eq:ode3} if $F \big( \bar{V}_1,\bar{V}_2,\bar{V}_3,\bar{V}_4\big)=(0,0,0,0)$.
We will classify these linear equilibria by their generalized eigenspaces, according to the sign of the real part of the eigenvalues of the linearization, giving the decomposition $\R^4=E^u\oplus E^s\oplus E^c$
into the direct sum of unstable, stable and center eigenspaces. 




\begin{lem} The points
\begin{equation}\label{eq:0:1}
\left(\begin{array}{c}0\\ 0\\ 0\\ 0\end{array}\right),\quad
\left(\begin{array}{c}1\\ 0\\ 0\\ 0\end{array}\right),\quad
\dfrac{1 -\tau_{h}/\tau_{hn}}{1\tau_{n} -1/\tau_{h}}\, \left(\begin{array}{c}
-1/\tau_{h}\\ 0\\ 1/\tau_{n} \\0
\end{array}\right),
\end{equation}
$\forall\tau_{n},\tau_{h},\tau_{hn}\in\R$ with $\tau_{n}\neq \tau_{h}$, are equilibrium points.
\end{lem}

\begin{rk}
For any $\tau_{n},\tau_{h},\tau_{hn}\in\R\setminus \{0\}$, the equilibrium $(\widetilde{V}_1, \widetilde{V}_2,\widetilde{V}_3, \widetilde{V}_4)^T \not\in \Big(\overline{\R_+}\Big)^4,\ $
\end{rk}

Let us denote by $A,\widehat{A},\widetilde{A}$ the Jacobian matrices at the equilibria, i.e.
$
A:=DF(0,0,0,0)$, $ \widehat{A}:=DF(1,0,0,0)$,
$\widetilde{A}:=DF(\widetilde{V}_1,\widetilde{V}_2,\widetilde{V}_3,\widetilde{V}_4)$. Let us denote their corresponding eigenvalues by
$\s(A):=\{\la_{1+},\la_{1-}, \la_{2+},\la_{2-}\}$ , $\s(\widehat{A}):=\{\widehat{\la}_{1+},
\widehat{\la}_{1-}, \widehat{\la}_{2+},\widehat{\la}_{2-}\}$, $\s(\widetilde{A})  = \{\widetilde{\la}_{1+},
\widetilde{\la}_{1-}, \widetilde{\la}_{2+},\widetilde{\la}_{2-}\}$


\begin{thm}\label{th:eq:00}
The point $\Big(\tau_{h};(V_1,V_2,V_3,V_4)\Big)
=\Big(\tau_{hn};(0,0,0,0)\Big)$ is a bifurcation point.


More precisely:

\begin{enumerate}
\item[\rm (a)]Assume that $\ c> c_{n}^*$.

\smallskip

\begin{enumerate}
\item[\rm (a.i)] If moreover
$\tau_{hn}>\tau_{h},$   then  $(0,0,0,0)$ is a stable equilibrium. Moreover

\smallskip

\begin{enumerate}
\item[\rm (a.i.1)] If  $c\geq c_{hn}^*$, then  $(0,0,0,0)$ is a stable node.

\smallskip

\item[\rm (a.i.2)] If  $c< c_{hn}^*$, then  $(0,0,0,0)$ is a stable node-focus.
\end{enumerate}

\smallskip

\item[\rm (a.ii)] If $\tau_{hn}<\tau_{h},$  then  $(0,0,0,0)$ is a  saddle point for any $c\in\R$.
\end{enumerate}

\medskip

\item[\rm (b)]Assume that $\ c< c_{n}^*.$

\smallskip

\begin{enumerate}
\item[\rm (b.i)] If moreover $\tau_{hn}>\tau_{h},$   then  $(0,0,0,0)$ is an stable equilibrium.

\smallskip

\begin{enumerate}
\item[\rm (b.i.1)] if  $c\geq c_{hn}^*$, then  $(0,0,0,0)$ is a stable focus-node.

\smallskip

\item[\rm (b.i.2)] If  $c< c_{hn}^*$, then  $(0,0,0,0)$ is a stable focus.
\end{enumerate}

\medskip

\item[\rm (b.ii)] If $\tau_{hn}<\tau_{h},$  then  $(0,0,0,0)$ is a  saddle point for any $c\in\R$.
\end{enumerate}
\end{enumerate}
\end{thm}

\begin{proof} We have to analyze the sign of the real part of the eigenvalues of $A$. Linearizing  $f$ and computing for $(0,0,0,0)$ we get


\begin{equation}\label{def:Abar}
A:=DF(0,0,0,0) =
\left(\begin{array}{cccc} 0 & 1 & 0  & 0\\
-\dfrac{1}{D_n\, \tau_{n}} & -\dfrac{{c}}{{D_n}} &  - \dfrac{1}{{D_n}\, {\tau_{hn}}}  & 0\\ 0 & 0 & 0 & 1\\0 & 0 & \dfrac{1}{D_h}\left( \dfrac{1}{\tau_{hn}}-\dfrac{1}{\tau_{h} }\right)  & -\dfrac{{c}}{{D_h}} \end{array}\right).
\end{equation}
its spectrum being given by
$$ P_0(\la) = \det (A-\la I)= \dfrac{1}{D_n\, D_h }\left(D_n\, \lambda^2 + c\, \lambda + \dfrac{1}{\tau_{n}}\right)\, \left(D_h\, \lambda^2 + c\, \lambda + \dfrac{1}{\tau_{h}} - \dfrac{1}{\tau_{hn}}\right).
$$
The eigenvalues can be written in the following way
\begin{equation}\label{def:eig}
\la_{1\pm} := \dfrac{-c\pm \sqrt{c^2-4D_n/\tau_{n}}}{2D_n},\qquad \la_{2\pm} := \dfrac{-c\pm\sqrt{c^2-4D_h\left(1/\tau_{h}- 1/\tau_{hn}\right)}}{2D_h}.
\end{equation}
Let us now analyze the sign of the real part of these eigenvalues

\begin{enumerate}
\item[\rm (a)] Obviously, if
$
c> c_{n}^*
$
then the eigenvalues $\la_{1\pm}$ are negative real numbers.

\smallskip

\begin{enumerate}
\item[\rm (a.i)] If moreover $\tau_{hn}>\tau_{h},$ then $c_{hn}^* \in \R$.

\smallskip

\begin{enumerate}
\item[\rm (a.i.1)] On the other side, if $c\geq c_{hn}^*,$ then the eigenvalues $\la_{2\pm}$ are negative real numbers and as a consequence,  all the eigenvalues of the Jacobian matrix are negative real numbers and the trivial equilibria is a estable node.

\smallskip

\item[\rm (a.i.2)] If  $c< c_{hn}^*$, then  $\la_{2\pm}\in\C^-.$
    As a conclusion, $\s(A)\subset \C^-,$ has two real eigenvalues and two complex eigenvalues, and $(0,0,0,0)$ is a stable node-focus.
\end{enumerate}

\smallskip

\item[\rm (a.ii)] If $\tau_{hn}<\tau_{h},$  then  $\la_{2-}<0<\la_{2+}.$
    Therefore, $\s(A)\subset \R,$ and has positive eigenvalue concluding that $(0,0,0,0)$ is a  saddle point for any $c\in\R$.
\end{enumerate}

\medskip

\item[\rm (b)] If $\ c< c_{n}^*,$  then  $\la_{1\pm}\in\C^-.$

\smallskip

\begin{enumerate}
\item[\rm (b.i)] If moreover $\tau_{hn}>\tau_{h},$   then  $c_{hn}^*\in\R.$

\smallskip

\begin{enumerate}
\item[\rm (b.i.1)] if  $c\geq c_{hn}^*$, then  $\s(A)\subset \C^-,$ and $(0,0,0,0)$ is a stable focus-node.

\smallskip

\item[\rm (b.i.2)] If  $c< c_{hn}^*$, then  $\la_{2\pm}\in\C^-,$   $\s(A)\subset \C^-,$ and $(0,0,0,0)$ is a stable focus.
\end{enumerate}

\medskip

\item[\rm (b.ii)] If $\tau_{hn}<\tau_{h},$  then  $\la_{2-}<0<\la_{2+}$,  and $(0,0,0,0)$ is a  saddle point for any $c\in\R$.

\end{enumerate}

\end{enumerate}
\end{proof}

\begin{rk}
Let us remark that if $c= c_{n}^*,$  then  $\la_{1+}=0$. Also 
if $\tau_{hn}=\tau_{h}$, then, by definition, $c_{hn}^*=0$, and $\la_{2+}=0$.
\end{rk}

\begin{thm}\label{th:eq:10}
For any $0<\tau_{n}<\infty$ and $\tau_{hn}>0,$ the equilibrium $(V_1,V_2,V_3,V_4)=(1,0,0,0)$ is a saddle point, $dim (E^u)=dim (E^s)=2$, and  $E^u$ is tangent to the space $span[(1,\widehat{\la}_{1+},0,0),(0,0,1,\widehat{\la}_{2+})]
$
at $(1,0,0,0)$.
\end{thm}
\begin{proof}

Since
$$
\widehat{A}=DF(1,0,0,0) =
\left(\begin{array}{cccc} 0 & 1 & 0 & 0\\
\dfrac{1}{D_n\, \tau_{n}}  & -\dfrac{{c}}{{D_n}}\ \ \ &
\dfrac{1}{D_n}\left(\dfrac{1}{\tau_{n}}-\dfrac1{\tau_{hn}}\right) & 0\\ 0 & 0 & 0 & 1\\
0 & 0 & \dfrac{1}{{D_h}\, {\tau_{hn}}}  & -\dfrac{{c}}{{D_h}} \end{array}\right),
$$
its spectrum is given by
$$
P_1(\la) =\dfrac{1}{D_n\, D_h}\left(D_n\, {\lambda}^2 + c\, \lambda - \dfrac{1}{\tau_{n}}\right)\, \left(D_h\, {\lambda}^2 + c\,  \lambda - \dfrac{1}{\tau_{hn}}\right).
$$
The eigenvalues are given by
\begin{equation}\label{def:eig:2}
\widehat{\la}_{1\pm} := \dfrac{-c\pm \sqrt{c^2+4D_n/\tau_{n}}}{2D_n},\qquad \widehat{\la}_{2\pm} := \dfrac{-c\pm\sqrt{c^2+4D_h/\tau_{hn}}}{2D_h}.
\end{equation}
Whenever  $\tau_{hn},\tau_{n}>0,$ the eigenvalues $\widehat{\la}_{1\pm},\ \widehat{\la}_{2\pm}\in\R$ and satisfy $
\widehat{\la}_{1-}<0<\widehat{\la}_{1+},\qquad \widehat{\la}_{2-}<0<\widehat{\la}_{2+},
$
and as a consequence,  the non-trivial equilibria $(1,0,0,0)^T$ is a saddle point, and  $dim (E^u)=dim (E^s)=2$. Moreover, it is not difficult to prove that $E^u$ is tangent to $span[(1,\widehat{\la}_{1+},0,0),$ $(0,0,1,\widehat{\la}_{2+})]$ at $(1,0,0,0)$, which completes the proof.
\end{proof}
\end{appendix}

\section*{Acknowledgements}

This work has been partially supported by the Ministerio de Econom\'{\i}a y Competitividad (Spain), under grants MTM2012-31073 and MTM2012-31298.

   \end{document}